\documentclass[10pt,conference]{IEEEtran}

\ifCLASSINFOpdf
  \usepackage[pdftex]{graphicx}
\else
\fi
\usepackage[cmex10]{amsmath}
\usepackage{amsthm,amsfonts,amssymb}
\usepackage{ifthen}
\usepackage{tikz}
\usepackage{cite}
\usepackage{amsthm,amsfonts,amssymb}
\usepackage{graphicx}
\usepackage{subfig}
\usepackage{blkarray}
\usepackage{dsfont}
\usepackage[mathscr]{euscript}
\usepackage{enumitem}

\usepackage[top=0.75in, bottom=0.75in, left=0.75in, right=0.75in]{geometry}
\usepackage{amsfonts}
\usepackage{mathrsfs}

\usepackage{algorithm}
\usepackage[noend]{algpseudocode}
\usepackage{blkarray}
\usepackage{stfloats}%
\newtheorem{theorem}{Theorem}[section]
\newtheorem{corollary}{Corollary}[section]
\newtheorem{proposition}{Proposition}[section]
\newtheorem{lemma}{Lemma}[section]
\theoremstyle{remark}
\newtheorem*{remark}{Remark}
\theoremstyle{definition}
\newtheorem{definition}{Definition}[section]

\usepackage{url}
\usepackage{tcolorbox}
\usepackage{arydshln}
\usepackage{mathtools}
\usepackage{dsfont}
\definecolor{brickred}{cmyk}{0,0.89,0.94,0.28}
\definecolor{goldenrod}{cmyk}{0,0.10,0.84,0}
\definecolor{purple}{cmyk}{0.45,0.86,0,0}
\definecolor{rawsienna}{cmyk}{0,0.72,1,0.45}
\definecolor{olivegreen}{cmyk}{0.64,0,0.95,0.40}
\definecolor{peach}{cmyk}{0,0.5,0.7,0}
\definecolor{darkolive}{rgb}{0.,0.4,0.}
\colorlet{grey}{gray!40}

\DeclareMathOperator{\argmin}{\text{argmin}}

\global\long\def\E{\mathbb{E}}

\usepackage{xcolor}

\usepackage{xcolor}

\usepackage[cmintegrals]{newtxmath}

\begin{document}
\title{Optimal Tree-Based Mechanisms for Differentially Private Approximate CDFs}

\author{\IEEEauthorblockN{
		V. Arvind Rameshwar,
		Anshoo Tandon,
		and
		Abhay Sharma
	}
	\thanks{The authors are with the India Urban Data Exchange Program Unit, Indian Institute of Science, Bengaluru, India, emails: \texttt{\{arvind.rameshwar, anshoo.tandon\}@gmail.com,  abhay.sharma@datakaveri.org}.}
}
\IEEEoverridecommandlockouts
\maketitle

\begin{abstract}

This paper considers the $\varepsilon$-differentially private (DP) release of an approximate cumulative distribution function (CDF) of the samples in a dataset. We assume that the true (approximate) CDF is obtained after lumping the data samples into a fixed number $K$ of bins. In this work, we extend the well-known binary tree mechanism to the class of \emph{level-uniform tree-based} mechanisms and identify $\varepsilon$-DP mechanisms that have a small $\ell_2$-error. We identify optimal or close-to-optimal tree structures when either of the parameters, which are the branching factors or the privacy budgets at each tree level, are given, and when the algorithm designer is free to choose both sets of parameters. Interestingly, when we allow the branching factors to take on real values, under certain mild restrictions, the optimal level-uniform tree-based mechanism is obtained by choosing equal branching factors \emph{independent} of $K$, and equal privacy budgets at all levels. Furthermore, for selected $K$ values, we explicitly identify the optimal \emph{integer} branching factors and tree height, assuming equal privacy budgets at all levels. Finally, we describe {general} strategies for improving the private CDF estimates further, by combining multiple noisy estimates and by post-processing the estimates for consistency.


\end{abstract}


%
\IEEEpeerreviewmaketitle

\section{Introduction}
It is now well-understood that the release of even seemingly innocuous functions of a dataset that is not publicly available can result in the reconstruction of the identities of individuals (or users) in the dataset with alarming levels of accuracy (see, e.g.,  \cite{narayanan, sweeney}). To alleviate concerns over such attacks, the framework of differential privacy (DP) was introduced in \cite{dwork06}, which guarantees the privacy of any single sample. Subsequently, several works (see the surveys \cite{dworkroth, vadhan2017} for references) have sought to design DP mechanisms or algorithms for the provably private release of statistics such as the mean, variance, counts, and histograms, resulting in the widespread adoption of DP for private data mining and analysis \cite{appledp,googledp}. 

In this work, we consider the fundamental problem of the DP release of (approximate) cumulative distribution functions (CDFs), via the release of a vector of suitably normalized, cumulative counts of data samples in bins of uniform width. We assume that the number of bins (or, equivalently, the length of the query output) is fixed to be some constant $K$. Such a setting has been well-studied in the context of the private release of histograms and interval query answers (see, e.g., \cite[Sec. 2.3]{vadhan2017} and the works \cite{barak,treedwork,treeshi,dworkrectangle}). While it is easy to extend the mechanisms constructed for private histogram release to the setting of (approximate) CDFs, there have been only few works that seek to optimize the parameters of such mechanisms to achieve low errors. In particular, the works \cite{qardaji, cormode} consider variants of the well-known binary tree mechanism and suggest choices of the tree branching factor that achieve low errors using somewhat unnatural error metrics and asymptotic analysis. On the other hand, in the context of continual counting, given a fixed choice of parameters of (variants of) the binary tree mechanism, the works \cite{honaker,pagh} suggest techniques to optimally process the information in the nodes of the tree and use multiple noisy estimates of the same counts to obtain low-variance estimates of interval queries. We mention also that there have been several works (see, e.g., \cite{matrix1,matrix2,matrix3} and references therein) on matrix factorization-based mechanisms that result in the overall optimal error for general ``linear queries'' (see \cite[Sec. 1.5]{vadhan2017} for the definition). In this work, we concentrate on the class of tree-based mechanisms and seek to optimize their parameters for low $\ell_2$-error.

We first revisit some simple mechanisms for differentially private CDF release, via direct interval queries or histogram-based approaches, and explicitly characterize their $\ell_2$-errors. While such results are well-known (see, e.g., \cite{treeshi}), they allow for comparisons with the errors of the broad class of ``level-uniform tree-based" mechanisms -- a class that we define in this work -- that subsumes the binary tree mechanism and its previously studied variants \cite{qardaji, cormode, pagh}. First, by relaxing the integer constraint on the branching factors of the tree, we identify the optimal mechanism within this class, which turns out to be a simple tree-based mechanism with equal branching factors and privacy budgets at all levels. Furthermore, for sufficiently large $K$, the optimal branching factor, under a mild restriction on the branching factors, is a \emph{constant} -- roughly $17$. We mention that, interestingly, \cite{qardaji} reports the optimal branching factor in a \emph{subclass} of mechanisms as $16$, but for a different error metric, and using somewhat imprecise analysis. Next, we reimpose the integer constraint on branching factors and identify the optimal tree structure when $K$ is a power of a prime, assuming that the privacy budgets at all levels of the tree are equal. We obtain this result as a by-product of more general results on the optimal heights of trees, for a broader class of $K$ values. Finally, we argue that for this setting where the branching factors are constrained to be integers, for selected values of $K$, {for tree-based mechanisms with equal privacy budgets across all levels,} those mechanisms with equal branching factors at all levels are in fact strictly sub-optimal.

We then describe a strategy for optimally post-processing noisy estimates of a CDF, in such a manner that the output is ``consistent'' \cite{barak,histhay,histml}, i.e., so that the output respects standard properties of a CDF. While post-processing for consistency allows for easy manipulations by downstream applications, we show that it has the additional benefit of resulting in decreased errors. Our solution makes use of a simple, dynamic programming-based approach for solving the consistency problem. In the appendix, we also revisit a well-known approach \cite{honaker} for refining private tree-based estimates of statistics by combining several noisy estimates together, in the context of private CDF release and discuss some analytical results.

Our work is of immediate use to theorists and practitioners alike, via its explicit identifications of easy-to-implement mechanisms and algorithms, for minimizing the $\ell_2$-error. 

\section{Notation and Preliminaries}
\label{sec:prelim}
\subsection{Notation}
The notation $\mathbb{N}$ denotes the set of positive natural numbers. For a given integer $n\in \mathbb{N}$, the notation $[n]$ denotes the set $\{1,2,\ldots,n\}$. An empty product is defined to be $1$. We use the notation $\mathbf{b}$ to denote a vector of predefined length, all of whose symbols are equal to $b\in \mathbb{R}$. Given a real-valued vector $\mathbf{a} = (a_1,a_2,\ldots,a_n)$, we denote its $\ell_1$-norm by $\lVert \mathbf{a} \rVert_1:= \sum_{i=1}^n |a_i|$ and its $\ell_2$-norm by $\lVert \mathbf{a} \rVert_2:= \left(\sum_{i=1}^n a_i^2\right)^{1/2}$. 
We use the notation $\text{Lap}(b)$ to refer to the zero-mean Laplace distribution with standard deviation $\sqrt{2}b$
; its probability distribution function (p.d.f.) obeys
$$
f(x) = \frac{1}{2b}e^{-|x|/b}, \ x\in \mathbb{R}.
$$
Given a random variable $X$, we denote its variance by $\text{Var}(X)$. The notation ``$\log$'' denotes the natural logarithm.

\subsection{Problem Formulation}
Consider a dataset $\mathcal{D}$ consisting of $N$ samples $x_1,\ldots,x_N$, such that $x_i\in \mathbb{R}$, $i\in [N]$. We assume further that each data sample is bounded, i.e., $x_i\in [a,b)$, for some $a,b \in \mathbb{R}$, for all $i\in [N]$. Fix a number of bins $K\in \mathbb{N}$, and define the $j^{\text{th}}$ bin $B_j$ to be the interval $\left[a+\frac{(j-1)(b-a)}{K},a+\frac{j(b-a)}{K}\right)$, for $1\leq j\leq K$. Now, for $j\in [K]$, let $c_j = c_j(\mathcal{D})$ denote the number of data samples $x_i$ in bin $B_j$, i.e.,
\[
c_j:= \left\lvert \{i:\ x_i\in B_j\}\right \rvert.
\]
Note that the vector $\mathbf{c} = \mathbf{c}(\mathcal{D}):= (c_1,\ldots,c_K)$ is a histogram corresponding to the dataset $\mathcal{D}$, with uniform bin width $1/K$. Let $\overline{c}_j = \overline{c}_j(\mathcal{D}):= \sum_{r\leq j} c_r$ denote the cumulative number of data samples in bins $B_1$ through $B_j$, for $j\in [K]$, with $\mathbf{\overline{c}}:= (\overline{c}_1,\ldots,\overline{c}_K)$. We are interested in the quantity
\begin{equation}
	\label{eq:cdf}
\mathbf{F} = \mathbf{F}(\mathcal{D}) = \frac{1}{N}\cdot \mathbf{\overline{c}},
\end{equation}
which can be interpreted as an ``approximate'' cumulative distribution function (CDF) corresponding to $\mathcal{D}$, with uniform bin width $1/K$. We make the natural assumption that the total number of samples $N$ in the dataset $\mathcal{D}$ of interest is known; such an assumption is implicitly made in most settings that employ the ``swap-model'' definition of neighbouring datasets in Section \ref{sec:dp} \cite{dwork06, vadhan2017} (as opposed to an ``add-remove'' model \cite[Def. 2.4]{dworkroth}, \cite{atsaddremove}). In this paper, we focus on releasing the length-$K$ vector $\mathbf{F}$ in a differentially private (DP) manner. In the next subsection, we revisit the notion of $\varepsilon$-differential privacy, under the ``swap-model'' definition of neighbouring datasets.



\subsection{Differential Privacy}
\label{sec:dp}
Consider datasets $\mathcal{D}_1 = (y_1,y_2,\ldots,y_n)$ and $\mathcal{D}_2 = ({y}_1',{y}_2',\ldots,{y}_n')$ consisting of the same number of data samples. Let $\mathsf{D}$ denote a universal set of such datasets. We say that $\mathcal{D}_1$ and $\mathcal{D}_2$ are ``neighbouring datasets'' if there exists $i\in [n]$ such that $y_i\neq {y}_i'$, with $y_j = {y}_j'$, for all $j\neq i$. In this work, we concentrate on mechanisms (or algorithms) that map a given dataset to a vector in $\mathbb{R}^K$, for some fixed $K$.

\begin{definition}
	For a fixed $\varepsilon>0$, a mechanism $M: \mathsf{D}\to \mathbb{R}^K$ is said to be $\varepsilon$-DP if for every pair of neighbouring datasets $\mathcal{D}_1, \mathcal{D}_2$ and for every measurable subset $Y \subseteq \mathbb{R}^K$, we have 
$
	\Pr[M(\mathcal{D}_1) \in Y] \leq e^\varepsilon\cdot \Pr[M(\mathcal{D}_2) \in Y].
$
\end{definition}
For any mechanism $M_g: \mathcal{D}\to \mathbb{R}^K$ that seeks to release a DP estimate of the function $g:\mathcal{D}\to \mathbb{R}^K$, we define the following error metrics.
\begin{definition}
	The $\ell_1$-error of $M_g$ on dataset $\mathcal{D}$ is defined to be $E_1(M_g):= \E\left[\left \lVert M_g(\mathcal{D}) - g(\mathcal{D})\right \rVert_1\right]$. Likewise, the squared $\ell_2$-error of $M_g$ on dataset $\mathcal{D}$ is defined to be $E_2(M_g):= \E\left[\left \lVert M_g(\mathcal{D}) - g(\mathcal{D})\right \rVert_2^2\right]$. In both cases, the expectation is over the randomness introduced by $M_g$.
\end{definition}
We remark that if $M_g$ is an unbiased estimator of $g$, the {expected} squared $\ell_{2}$-error of $M_g$ is precisely the variance of $M_g$. Next, we recall the definition of the sensitivity of a function.
\begin{definition}
	Given a function $g: \mathsf{D}\to \mathbb{R}^K$, we define its  sensitivity $\Delta_g$ as
$
	\Delta_g:= \max_{\mathcal{D}_1,\mathcal{D}_2\ \text{nbrs.}} \left \lVert g(\mathcal{D}_1) - g(\mathcal{D}_2)\right\rVert_1,
$
	where the maximization is over neighbouring datasets.
\end{definition}
{For example, the sensitivity of the function (or vector) $\mathbf{c}$ is $2$}, since changing a single data sample can, in the worst case, decrease the count $c_j$ by $1$ and increase $c_{j^\prime}$ by $1$, where $j'\neq j$ (see, e.g., \cite[Eg. 3.2]{dworkroth}). Furthermore, the sensitivity of $\overline{\mathbf{c}}$ is $K-1$, which is achieved when the value of a single sample in bin $B_1$ changes to a value that lies in bin $B_K$. Therefore, from \eqref{eq:cdf}, we have $$\Delta_{\mathbf{F}} = \frac{K-1}{N}.$$
 The next result is well-known \cite[Prop. 1]{dwork06}:

\begin{theorem}
	\label{thm:dp}
	For any $g: \mathsf{D}\to \mathbb{R}^K$, the mechanism $M^{\text{Lap}}_g: \mathsf{D}\to \mathbb{R}^K$ defined by
$$
	M^{\text{Lap}}_g(\mathcal{D}) = g(\mathcal{D})+\mathbf{Z},
$$
	where $\mathbf{Z} = (Z_1,\ldots,Z_K)$, with $Z_i\stackrel{\text{i.i.d.}}{\sim} \text{Lap}(\Delta_g/\varepsilon)$ is $\varepsilon$-DP.
\end{theorem}


\section{Simple Mechanisms for DP CDF Release}
In this section, we discuss two simple mechanisms that can be employed for releasing an estimate of $\mathbf{F}$ in an $\varepsilon$-DP manner. We then explicitly compute the variances (equivalently, the squared $\ell_2$-errors) of the (unbiased) estimators constructed by these mechanisms. 
While the results in this section are well-known (see, e.g., \cite[Sec. 3.1]{treeshi}), we restate them to serve as motivation for our discussion on the general class of tree-based mechanisms (defined in Section \ref{sec:tree}) and our subsequent characterization of the optimal mechanism within this broad class.

\subsection{DP CDFs Via DP Range Queries}
Fix an $\varepsilon>0$. Our first mechanism constructs an $\varepsilon$-DP estimate of $\mathbf{F}(\mathcal{D})$ by first constructing an $\varepsilon$-DP estimate of $\mathbf{\overline{c}}$, via an estimate of the vector $\overline{\mathbf{c}}$ defined earlier.  Queries for values of the form $\overline{c}_j$, for some $j\in [K]$, are also called as ``range'' or ``threshold'' queries (see, e.g., \cite[Sec. 2.1, Lec. 7]{privstat}). In particular, let the length-$K$ vector $M_{\overline{\textbf{c}}}^{\text{Lap}}(\mathcal{D}) = \left(M_{\overline{\textbf{c}},j}^{\text{Lap}}(\mathcal{D})\right)_{1\leq j\leq K}$ be such that
\[
M_{\overline{\textbf{c}},K}^{\text{Lap}}(\mathcal{D}) = N
\]
and
 $$M_{\overline{\textbf{c}},j}^{\text{Lap}}(\mathcal{D}) = \overline{c}_j+{Z}_j,$$ where {$1 \leq j \le K-1$, and} $Z_j\stackrel{\text{i.i.d.}}{\sim} \text{Lap}\left(\frac{K-1}{\varepsilon}\right)$. The mechanism then returns $\mathbf{F}^{\text{ind}}(\mathcal{D}) = \frac{1}{N}\cdot M_{\overline{\textbf{c}}}^{\text{Lap}}(\mathcal{D})$. It is easy to verify that $\mathbf{F}^\text{ind}(\mathcal{D})$ is an unbiased estimator of $\mathbf{F}(\mathcal{D})$. Furthermore, the following simple lemma holds:
\begin{lemma}
	\label{lem:ind}
	The mechanism $\mathbf{F}^\text{\normalfont ind}$ is $\varepsilon$-DP, and for all $\mathcal{D}$,
	\[
	E_2\left(\mathbf{F}^\text{\normalfont ind}(\mathcal{D})\right) = \frac{2(K-1)^3}{N^2\varepsilon^2}.
	\]
\end{lemma}
\begin{proof}
	The fact that $\mathbf{F}^\text{ind}$ is $\varepsilon$-DP follows from Theorem \ref{thm:dp} and the fact that (arbitrary) post-processing of an $\varepsilon$-DP mechanism is still $\varepsilon$-DP \cite[Prop. 2.1]{dworkroth}. Here, we make use of our observation that the sensitivity of $\overline{\mathbf{c}} = K-1$. Observe that \begin{align*}E_2\left(\mathbf{F}^\text{ind}(\mathcal{D})\right) &= \frac{1}{N^2}\cdot\sum_{j=1}^{K-1} E_2\left(\left(M_{\overline{\textbf{c}}}^{\text{Lap}}(\mathcal{D})\right)_j\right)\\ &= \frac{(K-1)\cdot \text{Var}(Z_j)}{N^2},\end{align*}
	by the unbiasedness of $M_{\overline{c}_j}^{\text{Lap}}(\mathcal{D})$. Noting that $\text{Var}(Z_j) = 2(K-1)^2/{\varepsilon}^2$, we obtain the statement of the lemma.
\end{proof}
\subsection{DP CDFs Via Histogram Queries} \label{sec:cdf_via_histogram}
Fix an $\varepsilon>0$. Our second mechanism constructs an $\varepsilon$-DP estimate of the histogram $\mathbf{c}$ and then processes this estimate to obtain an estimate of the CDF $\mathbf{F}$. In particular, let $M_{\mathbf{c}}^\text{Lap}(\mathcal{D}) = \left(M_{\mathbf{c},i}^\text{Lap}(\mathcal{D})\right)_{1\leq i\leq K-1}$ be a length-$(K-1)$ vector such that
\begin{equation}
\label{eq:mchihist}
M_{\mathbf{c},i}^\text{Lap}(\mathcal{D}) = {c}_i+{Z}_i^\prime,
\end{equation}
{for $1 \le i \le K-1$, with} $Z^\prime_i\stackrel{\text{i.i.d.}}{\sim} \text{Lap}(2/\varepsilon)$. 
The mechanism then returns $$\mathbf{F}^\text{hist}(\mathcal{D}) = \frac{1}{N}\cdot \overline{M}_{\mathbf{c}}^\text{Lap}(\mathcal{D}),$$
where $\overline{M}_{\mathbf{c}}^\text{Lap}(\mathcal{D})$ is a length-$K$ vector, with $$\overline{M}_{\mathbf{c},i}^\text{Lap}(\mathcal{D}) = \sum_{j\leq i} M_{\mathbf{c},i}^\text{Lap}(\mathcal{D})  = \overline{c}_i+\sum_{j\leq i} {Z^\prime_j},$$
for $1\le i\le K-1$, and $\overline{M}_{\mathbf{c},{\color{blue}K}}^\text{Lap}(\mathcal{D}):= N$.
 Again, one can verify that $\mathbf{F}^\text{hist}(\mathcal{D})$ is an unbiased estimator of $\mathbf{F}(\mathcal{D})$, with squared $\ell_2$-error as given in the following lemma:
\begin{lemma}
	\label{lem:varhist}
	The mechanism $\mathbf{F}^\text{\normalfont hist}$ is $\varepsilon$-DP, and for all $\mathcal{D}$,
	\[
	E_2\left(\mathbf{F}^\text{\normalfont hist}(\mathcal{D})\right) = \frac{4K(K-1)}{N^2\varepsilon^2}.
	\]
\end{lemma}
\begin{proof}
	The fact that $\mathbf{F}^\text{hist}$ is $\varepsilon$-DP follows since $M_{\mathbf{c}}^\text{Lap}$ is $\varepsilon$-DP and by the fact that (arbitrary) post-processing of an $\varepsilon$-DP mechanism is still $\varepsilon$-DP \cite[Prop. 2.1]{dworkroth}. Further, note that
	\begin{align*}
		E_2\left(\mathbf{F}^\text{hist}(\mathcal{D})\right)&= \frac{1}{N^2}\cdot\sum_{i=1}^{K-1} \E\left[\left(\sum_{j\leq i} Z_j^\prime\right)^2\right]\\
		&= \frac{1}{N^2}\cdot\sum_{i=1}^{K-1} \sum_{j\leq i} \E\left[\left(Z_j^\prime\right)^2\right]\\
		&= \frac{8}{N^2\varepsilon^2}\cdot \sum_{i=1}^{K-1} i = \frac{4K(K-1)}{N^2\varepsilon^2}.
	\end{align*}
The second inequality above holds since the $Z_j^\prime$, $j\in [K]$, are i.i.d. zero-mean random variables.
\end{proof}
From Lemmas \ref{lem:ind} and \ref{lem:varhist}, it is clear that for a fixed number $N$ of data samples, the squared $\ell_{2}$-error of $\mathbf{F}^\text{hist}$ is less than that of $\mathbf{F}^\text{ind}$, for $K \geq 4$. 

In the next section, we define another class of mechanisms, which we call the class of ``level-uniform tree-based" mechanisms, whose squared $\ell_{2}$-error is much lower than $\mathbf{F}^\text{ind}$ and $\mathbf{F}^\text{hist}$. This class of mechanisms includes the well-known binary tree mechanism \cite{treedwork,treeshi} as a special instance. However, as we show, the binary tree mechanism is not necessarily the optimal mechanism within the class of level-uniform tree-based mechanisms, in terms of {the} squared $\ell_2$-error, for the release of DP CDFs. We proceed to then explicitly identify the optimal mechanism, with some additional conditions in place, via a simple optimization procedure.
\section{Tree-Based DP CDFs Under the $\ell_{2}$-Error}
\label{sec:tree}
We now define the class of level-uniform tree-based mechanisms. In order to do so, we will require more notation. 

Fix a positive integer $m\geq 2$  and  positive integers $n_1,\ldots,n_{m-1}\geq 1$ such that $\prod_{i=1}^{m-1} n_i \geq K$. Now, consider a rooted tree $\mathcal{T}_m$ with $1+\sum_{i=1}^{m-1} \prod_{j=1}^i n_j$ nodes, which is organized into levels as follows. The topmost level, also called level $0$, consists of only the root node, designated as ${v}_1$. Each node $v_{1,j_1,\ldots,j_i}$ in level $i$, where $1\leq j_r\leq n_r$ (we set $j_0=n_0=1$), for $1\leq r \leq i$ and $1\leq i\leq m-1$, has $n_i$ children nodes at level $i+1$, which are designated as $v_{1,j_1,\ldots,j_i,1}$ through $v_{1,j_1,\ldots,j_i,n_{i}}$. The nodes at level $m-1$ are the leaf nodes of $\mathcal{T}_m$, which have no children. We call $m-1$ as the height of tree $\mathcal{T}_m$.

Now, given such a rooted tree $\mathcal{T}_m$, we associate with the root node ${v}_1$ the interval $[a,b)$. Given an interval $[\ell,L]$ associated with a node $v_{1,j_1,\ldots,j_i}$ in level $i$, with $0\leq i\leq m-1$, we associate with its child node $v_{1,j_1,\ldots,j_i,j_{i+1}}$ the interval $\left[\ell+\frac{{(L-\ell)}(j_{i+1}-1)}{n_i}, \ell+\frac{{(L-\ell)}j_{i+1}}{n_i}\right)$, where $1\leq j_{i+1}\leq n_i$. For ease of notation, we write $\mathcal{I}(v_{1,j_1,\ldots,j_i})$ to denote the interval associated with node $v_{1,j_1,\ldots,j_i}$ and $\chi(v_{1,j_1,\ldots,j_i})$ to denote the number of samples $x_r$ in the interval $\mathcal{I}(v_{1,j_1,\ldots,j_i})$, for $0\leq i\leq m-1$. Note that the intervals corresponding to leaf nodes  are of the form $\left[a+\frac{{(b-a)}(j-1)}{K},a+\frac{{(b-a)}j}{K}\right)$, for some $j\in [K]$. 

Finally, for every leaf node $v$ with $\mathcal{I}(v) = \left[a+\frac{{(b-a)}(j-1)}{K},a+\frac{{(b-a)}j}{K}\right)$, for some $j\in [K]$, we associate the ``cumulative'' interval $\overrightarrow{\mathcal{I}}(v):= \left[a,a+\frac{{(b-a)}j}{K}\right)$. Furthermore, for the leaf node $v$ with associated cumulative interval $\overrightarrow{\mathcal{I}}$, we define its so-called ``covering'' $$\mathcal{C}(\overrightarrow{\mathcal{I}}) = \mathcal{C}(v)$$ to be that collection obtained via the following procedure. 

We set $\mathcal{J} = \overrightarrow{\mathcal{I}}$ and $\mathcal{C}(v)$ {to be the empty set}. Starting from level $i=0$ (containing the root node), we identify the set $\mathcal{U}_i$ consisting of all nodes $u$ at level $i$ such that $\bigcup_{u\in \mathcal{U}_i} \mathcal{I}(u) \subseteq \mathcal{J}$, and update $\mathcal{C}(v) \gets \mathcal{C}(v)\, \bigcup \, \mathcal{U}_i$. We then set $\mathcal{J}\gets \mathcal{J}\setminus \bigcup_{u\in \mathcal{U}_i} \mathcal{I}(u)$ and $i\gets i+1$ and iterate until $\mathcal{J}$ is the empty set. The covering $\mathcal{C}(v)$ is hence the collection of nodes in $\mathcal{T}_m$ that ``covers'' $\overrightarrow{\mathcal{I}}$, starting at the root node and proceeding in a breadth-first fashion.

{Note that for the leaf node $v = v_{1,n_1,\ldots,n_{m-1}}$, we have  $\mathcal{C}(v) = v_1$.} Now, given a leaf node $v = v_{1,j_1,\ldots,j_{m-1}}$,   {with $v \neq v_{1,n_1,\ldots,n_{m-1}}$}, let $\tau = \tau(v)$ be the largest value $\ell\in [1:m-1]$ such that $j_\ell \neq n_\ell$. Further, for any $i\in [m-1]$, we write $v_{1,j_1,\ldots,j_{i-1},<j_{i}}$ as shorthand for the nodes $v_{1,j_1,\ldots,j_{i-1},1},\ldots, v_{1,j_1,\ldots,j_{i-1},j_i-1}$. The following simple lemma then holds, by virtue of our convenient notation:

\begin{lemma}
	\label{lem:tau}
	Let $\overrightarrow{\mathcal{I}}$ be a cumulative interval associated with a leaf node $v = v_{1,j_1,\ldots,j_{m-1}}$. If $\overrightarrow{\mathcal{I}} = [a,b)$, then $\mathcal{C}(\overrightarrow{\mathcal{I}}) = v_1$. Else, 
	\[
	\mathcal{C}(\overrightarrow{\mathcal{I}}) = \{v_{1,<j_1},v_{1,j_1,<j_2},\ldots,v_{1,j_1,j_2,\ldots,<j_{\tau}}, v_{1,j_1,j_2,\ldots,j_\tau}\},
	\]
	where $\tau = \tau(v)$.
\end{lemma}

A (level-uniform) tree-based mechanism for DP CDF release constructs an estimate of the vector of counts $\mathbf{\overline{c}} = \mathbf{\overline{c}}(\mathcal{D})$ by first designing an unbiased estimator of the counts $\chi(v_{1,j_1,\ldots,j_i})$, for each node $v_{1,j_1,\ldots,j_i}$ in $\mathcal{T}_m$, where $0\leq i\leq m-1$ and $1\leq j_r\leq n_r$, for $1\leq r\leq i$. In particular, let $\varepsilon_1,\varepsilon_2,\ldots,\varepsilon_{m-1}$ be a collection of ``privacy budgets,'' with $\varepsilon_i$ assigned to the nodes at level $i$, $1\leq i\leq m-1$, with $\sum_{i=1}^{m-1}\varepsilon_i = \varepsilon$. Note here that we assign no privacy budget for the root node $v_1$, since its count $\chi(v_1) := N$ is assumed to be publicly known. Now, for each node $v = v_{1,j_1,\ldots,j_i}$ at level $i\in [m-1]$, we set
\begin{equation}
	\label{eq:mchi}
M^\text{Lap}_{\chi(v)}(\mathcal{D}) = \chi(v)+{Z}_{v},
\end{equation}
where ${Z}_v\stackrel{\text{i.i.d.}}{\sim} \text{Lap}\left(\frac{2}{{\varepsilon_i}}\right)$. We let $M^\text{Lap}_{\chi(v_1)}(\mathcal{D}) = N$. We then define
\begin{equation}
	\label{eq:mchilap}
	M^\text{Lap}_{\chi}(\mathcal{D}) := \left(M^\text{Lap}_{\chi(v)}(\mathcal{D}):\ v\ \text{is a node of } \mathcal{T}_m\right).\end{equation}
\begin{definition}
	\label{def:cdftree}
	For a fixed $\varepsilon>0$ and a positive integer $K$, consider a positive integer $m\geq 2$ (where $m-1$ represents the height of the tree) and a vector of positive integers $\mathbf{n}:=(n_1,\ldots,n_{m-1})$ (representing the branching factors at each level of the tree), with $n_i\geq 2$, for all $i\in [m-1]$, and a vector of positive real numbers $\boldsymbol{\varepsilon}:= (\varepsilon_1,\ldots,\varepsilon_{m-1})$, such that $\prod_{i=1}^{m-1} n_i \geq K$ and $\sum_{i=1}^{m-1}\varepsilon_i = \varepsilon$. We define  a level-uniform tree-based mechanism $\mathbf{F}^\text{tree}_{\mathbf{n},\boldsymbol{\varepsilon}}$ as one that obeys
	\[
	\mathbf{F}^\text{tree}_{\mathbf{n},\boldsymbol{\varepsilon}}(\mathcal{D}) = \frac{1}{N}\cdot \overline{M}^\text{Lap}_{\chi}(\mathcal{D}).
	\]
	Here, $\overline{M}^\text{Lap}_{\chi}$ is a length-$K$ vector indexed by leaf nodes $v = v_{1,j_1,\ldots,j_{m-1}}$ such that 
	\begin{equation}
		\label{eq:cdfcompute}
	\overline{M}^\text{Lap}_{\chi, v}(\mathcal{D}) := \sum_{v^\prime\in \mathcal{C}(v)} M^\text{Lap}_{\chi(v^\prime)}(\mathcal{D}).
	\end{equation}
	
	The class $\mathcal{F}^\text{tree}_m$ of level-uniform tree-based mechanisms consists of all mechanisms $\mathbf{F}^\text{tree}_{\mathbf{n},\boldsymbol{\varepsilon}}$ with tree height $m-1$, parameterized by integers $n_1,\ldots,n_{m-1}$ and real numbers $\varepsilon_1,\ldots,\varepsilon_{m-1}$ as above. We let $\mathcal{F}^\text{tree} = \bigcup_{m\geq 2} \mathcal{F}^\text{tree}_m$.
\end{definition}
\begin{remark}
	Clearly, by setting $n_1 = n_2 = \ldots = n_{m-1} = 2$, we obtain the well-studied binary tree mechanism \cite{treedwork,treeshi}. {Similarly, by setting $m=2$, {giving rise to $n_1=K$ and $\varepsilon_1 = \varepsilon$}, we recover the CDF-via-histogram mechanism (see Sec.~\ref{sec:cdf_via_histogram}), i.e., $\mathbf{F}^\text{tree}_{\mathbf{n},\boldsymbol{\varepsilon}}(\mathcal{D}) = \mathbf{F}^\text{hist}(\mathcal{D}) $.} We mention that much larger classes of tree-based mechanisms can be constructed by allowing the number of children and the privacy budgets of nodes at the same level to be different; we however do not consider such mechanisms in this work.
\end{remark}

Further, we mention that as a possible extension of the level-uniform tree-based mechanisms defined above, one can potentially design mechanisms with different privacy allocations $\{\varepsilon_{i,v}\}$, for each node $v$ at level $i\in [m-1]$. However, from the argument that follows, we see that in the simple setting of private histogram release, such a non-uniform allocation of privacy budgets across bins can do no better than equal privacy budget allocations for each bin, in terms of squared $\ell_2$-error. Formally, a mechanism $\tilde{M}_{\mathbf{c}}^\text{Lap}(\mathcal{D})$ can be constructed with $\boldsymbol{\varepsilon} = \left(\varepsilon_{i}:\ i\in [K-1]\right)$ being the vector of \emph{non-level-uniform} privacy allocations across bins, which are constrained to make the overall mechanism $\varepsilon$-DP, as will be elaborated next. Such a mechanism is defined the exact same way as in \eqref{eq:mchihist}, with the only difference being that the random variable $Z_i'$ is now drawn (independent of $Z_j'$, $j\neq i$) from the Lap$\left(\frac{1}{\varepsilon_{i}}\right)$ distribution. It can be argued that the overall mechanism is thus $\left(\max\limits_{i\neq j} \left(\varepsilon_i+\varepsilon_j\right)\right)$-DP. The (non-level-uniform) privacy allocations are hence chosen so that $\varepsilon = \max\limits_{i\neq j}\left(\varepsilon_i+\varepsilon_j\right)$. It can then be 
 argued that $E_2\left(\tilde{M}_{\mathbf{c}}^\text{Lap}(\mathcal{D})\right)$ is minimized by choosing $\varepsilon_{i} = \varepsilon/2$, for all $i\in [K-1]$; we omit a formal proof since it is outside the scope of this paper. Following such an intuition, we restrict our attention to level-uniform tree-based mechanisms in this paper.
 
The following fact then holds, for reasons similar to those in \cite{treedwork,treeshi}.
\begin{proposition}
	Any mechanism $\mathbf{F}^\text{\normalfont tree}_{\mathbf{n},\boldsymbol{\varepsilon}}\in \mathcal{F}^\text{\normalfont tree}$ is $\varepsilon$-DP.
\end{proposition}
In what follows, we exactly characterize the $\ell_2$-error of any mechanism $\mathbf{F}^\text{tree}_{\mathbf{n},\boldsymbol{\varepsilon}}\in \mathcal{F}^\text{tree}_m$. This then naturally leads to the identification of the parameters $n_1,\ldots,n_{m-1}$ and $\varepsilon_1,\ldots,\varepsilon_{m-1}$ that minimize an upper bound on the $\ell_2$-error, for fixed $K$ and $m$. We then optimize the $\ell_2$-error thus obtained over the branching factors, which in turn helps us identify an optimal tree-based mechanism $\textbf{F}^\text{tree}\in \mathcal{F}^\text{tree}$. For the purpose of the optimization procedures in this section, we relax $n_1,\ldots,n_{m-1}$ to be positive reals, instead of integers. In Section \ref{sec:integer}, we shall discuss some strategies for choosing good \emph{integer} parameters $n_1,\ldots,n_{m-1}$ that give rise to low $\ell_2$-error. 

\begin{theorem}
	\label{thm:e2tree}
	Given positive reals $n_1,\ldots,n_{m-1}$ such that $\prod_{i=1}^{m-1} n_i = K$, we have
	\[
	E_2\left(\mathbf{F}^\text{\normalfont tree}_{\mathbf{n},\boldsymbol{\varepsilon}}(\mathcal{D})\right) = \frac{4K}{N^2}\cdot  \sum_{i=1}^{m-1}\left(\frac{n_i-1}{\varepsilon_i^2}\right).
	\]
\end{theorem}
\begin{proof}
	We first observe that $$E_2\left(\mathbf{F}^\text{tree}_{\mathbf{n},\boldsymbol{\varepsilon}}(\mathcal{D})\right) =\frac{1}{N^2}\cdot\sum_{\text{\normalfont leaf nodes $v$}} E_2\left(\overline{M}^\text{Lap}_{\chi, v}(\mathcal{D})\right).$$ Now, consider any term $E_2(\overline{M}^\text{Lap}_{\chi, v})$, for a leaf node $v$. Note that
	\begin{align*}
	E_2\left(\overline{M}^\text{Lap}_{\chi, v}(\mathcal{D})\right) &= \text{Var}\left(\sum_{v^\prime\in \mathcal{C}(v)} {Z}_{v^\prime}\right)= \sum_{v^\prime\in \mathcal{C}(v)} \text{Var}\left({Z}_{v^\prime}\right),
	\end{align*}
by the independence of the $Z_{v^\prime}$ random variables. Hence, we have
\begin{equation} \label{eq:error1}E_2\left(\mathbf{F}^\text{tree}_{\mathbf{n},\boldsymbol{\varepsilon}}(\mathcal{D})\right) = \frac{1}{N^2}\cdot \sum_{\text{\normalfont leaf nodes $v$}} \sum_{v^\prime\in \mathcal{C}(v)} \text{Var}\left({Z}_{v^\prime}\right).\end{equation}
	In the summation in \eqref{eq:error1} above, let $s_u$ be the number of times any (non-root) node $u = v_{1,j_1,j_2,\ldots,j_i}$, for some $i\in [m-1]$, appears in the summation. Note that, by the assumption that the value $N$ is publicly known, we have $\text{Var}(Z_{v_1}) = 0$, and hence we set $s_{v_1} = 0$. Let $P(u)$ denote the (unique) parent of $u$ in $\mathcal{T}_m$ and let $\eta(u)$ denote the number of leaves of the subtree with $u$ as the root. By the structure of the sets $\mathcal{C}(v)$ for leaf nodes $v$, we have that \begin{align*}
		s_u &= \eta(P(u)) - \sum_{\stackrel{v = v_{1,j_1,\ldots,j_{i-1}, r_i}:}{\ r_i\leq j_i}}  \eta(v)\\
		&= \prod_{t=i}^{m-1} n_t - \sum_{r=1}^{j_i} \prod_{t=i+1}^{m-1}n_t\\
		&=\prod_{t=i}^{m-1} n_t - j_i\cdot \prod_{t=i+1}^{m-1}n_t = (n_i-j_i) \cdot \prod_{t=i+1}^{m-1}n_t.
	\end{align*}
	Hence, we have that 
	\begin{align*}
		N^2\cdot E_2(\mathbf{F}^\text{tree}_{\mathbf{n},\boldsymbol{\varepsilon}}(\mathcal{D}))  &= \sum_{u\text{ in $\mathcal{T}_m$}} s_u\cdot \text{Var}(Z_u)\\
		&= \sum_{i=1}^{m-1}\frac{8}{\varepsilon_i^2}
		\cdot \prod_{r=1}^{i-1} n_r\cdot \sum_{j_i = 1}^{n_i} \left((n_i-j_i) \prod_{t=i+1}^{m-1}n_t \right)\\
		&= \sum_{i=1}^{m-1}\frac{8}{\varepsilon_i^2}\cdot \prod_{t\neq i}n_t \sum_{j_i = 1}^{n_i}(n_i-j_i)\\
		&= 4 \prod_{t}n_t \cdot  \sum_{i=1}^{m-1}\left(\frac{n_i-1}{\varepsilon_i^2}\right)\\
		&= 4K\cdot  \sum_{i=1}^{m-1}\left(\frac{n_i-1}{\varepsilon_i^2}\right),
	\end{align*}
thus concluding the proof.
\end{proof}
\begin{remark}
Consider the estimation of the (approximate) CDF of a dataset $\mathcal{D}$ in the user-level differential privacy \cite{userlevel} setting, where each user $i\in [N]$ contributes $L\geq 1$ samples; here, the total number of samples in $\mathcal{D}$ is $NL$. For a fixed number of bins $K$, the parameters of a level-uniform tree-based mechanism $\mathbf{F}_{\mathbf{n},\boldsymbol{\varepsilon}}^\text{tree}$ must obey $\prod_{i\leq m-1} n_i = K$ and $L\cdot\sum_{i\leq m-1} \varepsilon_i = \varepsilon$, where the latter constraint ensures user-level privacy when potentially all of a user's samples are changed. Now, since in this setting the normalization in \eqref{eq:cdf} is by $NL$ instead of $N$, one obtains the \emph{same} expression for the squared $\ell_2$-error of $\mathbf{F}_{\mathbf{n},\boldsymbol{\varepsilon}}^\text{tree}$ as in Theorem \ref{thm:e2tree}. This indicates that standard techniques for trading off estimation error due to privacy for bias induced by clipping the number of samples contributed by each user \cite{userlevel,arat2024tit} will not be of use in this particular setting {where each user contributes \emph{equal} number of samples}.
\end{remark}

Theorem \ref{thm:e2tree} thus gives rise to a simple procedure for identifying an ``optimal'' tree-structure $\{n_i\}_{i\leq m-1}, \{\varepsilon_i\}_{i\leq m-1}$, that minimizes $E_2(\mathbf{F}^\text{tree}_{\mathbf{n},\boldsymbol{\varepsilon}})$, for given values of $\mathbf{n}$ or $\boldsymbol{\varepsilon}$. We recall that we relax the vector $\mathbf{n}$ to now be a vector of positive reals, instead of positive integers.
\begin{theorem}
	\label{thm:e2treeopt1}
	Given a vector of integers $\overline{\textbf{n}} = (\overline{n}_1,\ldots,\overline{n}_{m-1})\geq 0$ such that $\prod_{i\leq m-1} \overline{n}_i = K$, an optimal choice of positive reals $\boldsymbol{\varepsilon}:= (\varepsilon_1,\ldots,\varepsilon_{m-1})$ satisfying $\sum\limits_{i=1}^{m-1} \varepsilon_i = \varepsilon$ that minimizes $E_2(\mathbf{F}^\text{\normalfont tree}(\overline{\mathbf{n}},\boldsymbol{\varepsilon};\mathcal{D}))$, is given by
	\[
    \varepsilon_i = \varepsilon\cdot \frac{(\overline{n}_i-1)^{1/3}}{\sum\limits_{j=1}^{m-1}(\overline{n}_j-1)^{1/3}}.
 	\]
	Likewise, given $\overline{\boldsymbol{\varepsilon}} = (\overline{\varepsilon}_1,\ldots,\overline{\varepsilon}_{m-1})\geq 0$ such that $\sum_{i\leq m-1} \overline{\varepsilon}_i = \varepsilon$, an optimal choice of positive reals $\textbf{n}:= (n_1,\ldots,n_{m-1})$, satisfying $\prod\limits_{i=1}^{m-1} n_i = K$, that minimizes $E_2(\mathbf{F}^\text{\normalfont tree}(\mathbf{n},\overline{\boldsymbol{\varepsilon}};\mathcal{D}))$ is given by
	\[
	n_i = K^{1/(m-1)} \cdot \frac{\overline{\varepsilon}_i^2}{\left(\prod\limits_{j=1}^{m-1} \overline{\varepsilon}_j^2\right)^{1/(m-1)}}.
	\]
\end{theorem}
\begin{proof}
	Consider first the case when the vector of branching factors is fixed to be some $\overline{\mathbf{n}}$. From Theorem \ref{thm:e2tree}, we see that we simply wish to perform the following optimization procedure:
	\begin{align}
		&{\text{minimize}}\quad f(\boldsymbol{\varepsilon}):= \sum_{i=1}^{m-1} \frac{\overline{n_i}-1}{\varepsilon_i^2}\notag\\
		&\text{subject to:}\ \  \sum_{i=1}^{m-1} \varepsilon_i = \varepsilon. \tag{O$^1$} \label{eq:opt1}
	\end{align}
It can easily be seen that the function $f$ above is convex in $\boldsymbol{\varepsilon}$. By standard arguments in convex optimization theory, the KKT conditions are necessary and sufficient for optimality \cite[Sec. 5.5.3]{boyd} of a solution to (O$^1$). We hence obtain that, for all $i\in [m-1]$, we must have, for optimality,
\[
(\overline{n_i}-1)\cdot \varepsilon_i^{-3} = \lambda,
\]
for some constant $\lambda\geq 0$. Enforcing the condition that $\sum_{i\leq m-1} \varepsilon_i = \varepsilon$, and solving for $\lambda$ results in the value of $\varepsilon_i$ as in the theorem.

Next, consider the case when the vector of privacy budgets is fixed to be some $\boldsymbol{\varepsilon}$. Let us change variables to set $r_i:= \log n_i$, for all $i\in [m-1]$, with $\mathbf{r}:= (r_1,\ldots,r_{m-1})$. By the monotonicity of the exponential function, we see that our objective reduces to solving
	\begin{align}
	&{\text{minimize}}\quad g(\mathbf{r}):= \sum_{i=1}^{m-1} \frac{e^{r_i}-1}{\overline{\varepsilon_i}^2}\notag\\
	&\text{subject to:}\ \  \sum_{i=1}^{m-1} r_i = \log K. \tag{O$^2$} \label{eq:opt2}
\end{align}
Again, owing to the convexity of $g$ in $\mathbf{r}$, via the KKT optimality conditions, any optimal choice of $\{r_i\}$ values must satisfy
\[
e^{r_i} = \lambda' \cdot \overline{\varepsilon}_i^2,
\]
for some constant $\lambda'\geq 0$. Enforcing the condition that $\sum_{i\leq m-1} r_i = \log K$, and solving for $\lambda'$ results in the value of $n_i = e^{r_i}$ as stated in the theorem.
\end{proof}
Now, from Theorem \ref{thm:e2treeopt1}, we see that the function
\[
\overline{E}_2(\mathbf{F}^\text{tree}_{\mathbf{n},\boldsymbol{\varepsilon}}(\mathcal{D})):= \frac{4K}{N^2}\cdot  \sum_{i=1}^{m-1}\frac{n_i}{\varepsilon_i^2}
\]
satisfies
\[
\frac12\cdot\overline{E}_2(\mathbf{F}^\text{tree}_{\mathbf{n},\boldsymbol{\varepsilon}}(\mathcal{D}))\leq{E}_2(\mathbf{F}^\text{tree}_{\mathbf{n},\boldsymbol{\varepsilon}}(\mathcal{D}))\leq\overline{E}_2(\mathbf{F}^\text{tree}_{\mathbf{n},\boldsymbol{\varepsilon}}(\mathcal{D})),
\]
for all admissible branching factors $\mathbf{n}$ and privacy budgets $\boldsymbol{\varepsilon}$ such that $\prod_i n_i = K$ and $\sum_i \varepsilon_i = \varepsilon$, with $n_i\geq 2$ and $\varepsilon_i\geq 0$, for all $i\in [m-1]$. Thus, $\overline{E}_2(\mathbf{F}^\text{tree}_{\mathbf{n},\boldsymbol{\varepsilon}}(\mathcal{D}))$ serves as a good approximation for ${E}_2(\mathbf{F}^\text{tree}_{\mathbf{n},\boldsymbol{\varepsilon}}(\mathcal{D}))$. The following simple lemma then identifies the optimal tree structure that minimizes $\overline{E}_2(\mathbf{F}^\text{tree}_{\mathbf{n},\boldsymbol{\varepsilon}}(\mathcal{D}))$.

\begin{lemma}
	\label{lem:e2treeoptapprox}
	For a fixed integer $m\geq 2$, we have that the minimum of 
	$\overline{E}_2(\mathbf{F}^\text{\normalfont tree}_{\mathbf{n},\boldsymbol{\varepsilon}}(\mathcal{D}))
	$ over all admissible $\mathbf{n}, \boldsymbol{\varepsilon}$ 
	is achieved by setting $n_i = \beta := K^{\frac{1}{m-1}}$  and $\varepsilon_i = \widehat{\varepsilon}:= \frac{\varepsilon}{m-1}$, for all $i\in [m-1]$.
\end{lemma}
\begin{proof}
	Observe that
	\begin{align*}
		\overline{E}_2(\mathbf{F}^\text{tree}_{\mathbf{n},\boldsymbol{\varepsilon}}(\mathcal{D}))
		&\geq \frac{4K(m-1)}{N^2}\cdot \left(\prod_{i\leq m-1}\frac{n_i}{\varepsilon_i^2}\right)^{1/(m-1)}\\
		&= \frac{4K^{m/(m-1)}(m-1)}{N^2}\cdot \left(\frac{1}{\left(\prod_{i\leq m-1}\varepsilon_i\right)^{2/(m-1)}}\right)\\
		&\geq \frac{4K^{m/(m-1)}(m-1)^3}{N^2\cdot \left(\sum_{i\leq m-1}\varepsilon_i\right)^{2}}= \frac{4K^{m/(m-1)}(m-1)^3}{N^2 \varepsilon^2}.
	\end{align*}
Here, both the inequalities are due to the AM-GM inequality; they are met with equality if and only if the $n_i$ values are equal, for all $i$, and the $\varepsilon_i$ values are equal, for all $i$. Imposing the constraint that $\prod_i n_i = K$ and $\sum_i \varepsilon_i = \varepsilon$ yields the values $\beta, \widehat{\varepsilon}$.
\end{proof}

\begin{figure*}
	\centering
	\includegraphics[width = 0.8\linewidth]{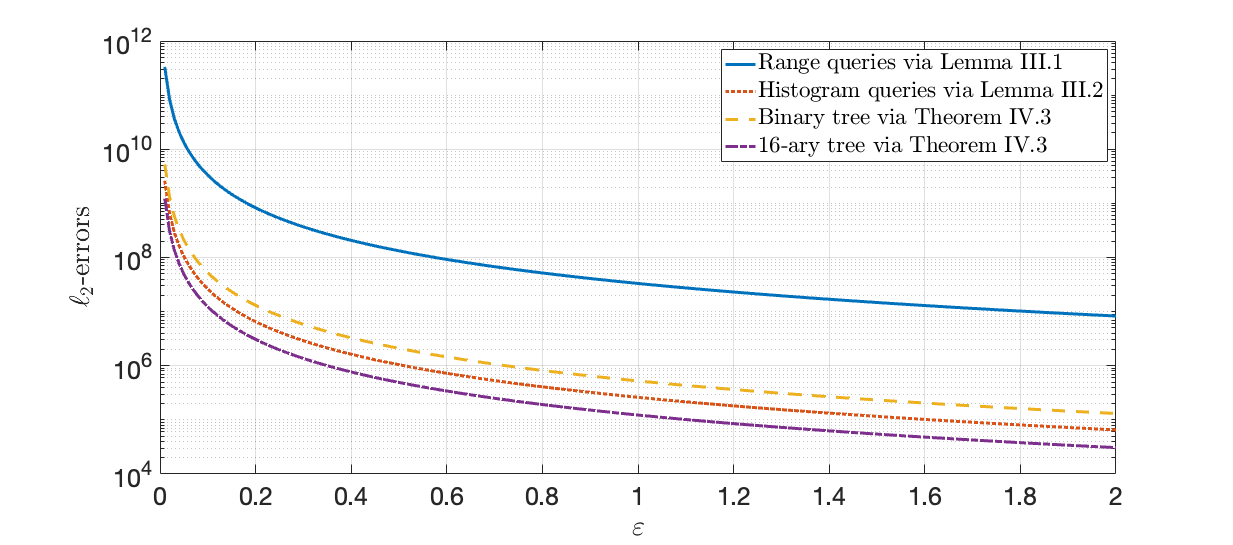}
	\caption{Plots of the normalized $\ell_2$-errors of the mechanisms discussed in the paper, for $K = 256$ and varying $\varepsilon$. The $\ell_2$-error axis is shown on a log-scale.}
	\label{fig:plots}
\end{figure*}

If we impose the additional constraint that $n_i\geq 3$, for all $i\in [m-1]$, then, Theorem \ref{thm:e2treeopt2} below tells us that the optimal tree structure that minimizes the true squared $\ell_2$-error ${E}_2(\mathbf{F}^\text{tree}_{\mathbf{n},\boldsymbol{\varepsilon}}(\mathcal{D}))$ also has equal branching factors and equal privacy budgets at all levels. For a given integer $m\geq 2$ corresponding to a tree $\mathcal{T}_m$, let $E_2^\text{OPT}$ denote this smallest squared $\ell_2$-error over all level-uniform tree-based mechanisms on $\mathcal{T}_m$ with branching factors $\mathbf{n}$ such that $n_i\geq 3$, for all $i\in [m-1]$.
\begin{theorem}
	\label{thm:e2treeopt2}
	For a fixed integer $m\geq 2$, we have that the minimum of $E_2(\mathbf{F}^\text{\normalfont tree}_{\mathbf{n},\boldsymbol{\varepsilon}}(\mathcal{D}))
	$ over all admissible $\mathbf{n}, \boldsymbol{\varepsilon}$, with $n_i\geq 3$, for all $i\in [m-1]$, is achieved by setting $n_i = \beta := K^{\frac{1}{m-1}}$  and $\varepsilon_i = \widehat{\varepsilon}:= \frac{\varepsilon}{m-1}$, for all $i\in [m-1]$.
	
Furthermore, for this choice of $\mathbf{n}, \boldsymbol{\varepsilon}$ values, the optimal $\ell_2$-error of any level-uniform tree-based mechanism is
\begin{align*}
E_2^{\text{\normalfont OPT}} &= \frac{4K(m-1)^3}{N^2\varepsilon^2} \cdot \left(-1+{K^{\frac{1}{m-1}}}\right)\\
&= \frac{4K(\log K)^3}{N^2\varepsilon^2}\cdot \left(\frac{\beta-1}{(\log \beta)^3}\right).
\end{align*}
\end{theorem}

\begin{proof}
	From Theorem \ref{thm:e2tree},we see that it suffices to consider the following optimization problem:
	\begin{align}
		&{\text{minimize}}\quad \sum_{i=1}^{m-1} \frac{n_i-1}{{\varepsilon_i}^2}\notag\\
		&\text{subject to:}\ \  \sum_{i\leq m-1} \varepsilon_i = \varepsilon\ \text{and}\ \prod_{i\leq m-1} n_i = K. \tag{O} \label{eq:optfin1}
	\end{align}
As in the proof of Theorem \ref{thm:e2treeopt1}, we first change variables and let $r_i:= \log n_i$, for all $i\in [m-1]$, with $\mathbf{r}:= (r_1,\ldots,r_{m-1})$. The optimization problem (O) then reduces to:
\begin{align}
	&{\text{minimize}}\quad \Phi(\mathbf{r},\boldsymbol{\varepsilon}):= \sum_{i=1}^{m-1} \frac{e^{r_i}-1}{{\varepsilon_i}^2}\notag\\
	&\text{subject to:}\ \  \sum_{i\leq m-1} \varepsilon_i = \varepsilon\ \text{and}\ \sum_{i\leq m-1} r_i = \log K. \tag{O'} \label{eq:optfin2}
\end{align}
	In Appendix \ref{sec:app-convex}, we prove that $\Phi$ is jointly convex in $\mathbf{r}, \boldsymbol{\varepsilon}$, when $r_i\geq \log 3$ and $\varepsilon> 0$. Then, via the KKT optimality conditions \cite[Sec. 5.5.3]{boyd}, any optimal $(\mathbf{r},\boldsymbol{\varepsilon})$ must obey
	\[
	\frac{e^{r_i}}{\varepsilon_i^2} = \mu\ \text{and}\ \frac{e^{r_i}-1}{\varepsilon_i^3} = \zeta,
	\]
	for some constants $\mu, \zeta\geq 0$. Clearly, the choice $r_1 = \ldots = r_{m-1} = \frac{\log K}{m-1}$ and $\varepsilon_1= \ldots = \varepsilon_{m-1} = \frac{\varepsilon}{m-1}$ satisfies these requirements. The expression for the $\ell_{2}$-error then holds from Theorem \ref{thm:e2tree}.
\end{proof}
The expression for the optimal $\ell_{2}$-error in Theorem \ref{thm:e2treeopt2} also allows us to explicitly identify that value of $\beta${, denoting the common branching factor at each level}, that minimizes $E_2^\text{OPT}$. Note that the Lemma below, unlike Theorem \ref{thm:e2treeopt2}, does not assume that $m$ is fixed; instead, it is implicitly assumed that $m = \log_\beta(K)+1$.
\begin{lemma}
	\label{lem:optbranch}
	The expression
	\[
	E_2^{\text{\normalfont OPT}} = E_2^{\text{\normalfont OPT}}(\beta) = \frac{4K(\log K)^3}{N^2\varepsilon^2}\cdot \left(\frac{\beta-1}{(\log \beta)^3}\right)
	\]
	is minimized over $\beta\geq 2$ by $\beta^\star$ that satisfies $\log \beta^\star = 3\cdot \left(\frac{\beta^\star-1}{\beta^\star}\right)$. In particular, the $E_2^{\text{\normalfont OPT}}(\beta)$ is minimized over the integers by $\beta^\star_\text{int} = 17$.
\end{lemma}
\begin{proof}
	Observe that minimizing $E_2^{\text{OPT}}(\beta)$ over $\beta$ is equivalent to minimizing $$f(\beta) = \frac{\beta-1}{(\log \beta)^3}.$$
    Standard arguments show that the derivative $f'(\beta)$ is non-positive, for $\beta\leq \beta^\star$, and is positive, for $\beta>\beta^\star$, where $\beta^\star$ is as in the statement of the lemma. Here, we make use of the fact that the solution $\beta^\star$ to the equation $\log \beta = 3\cdot \left(\frac{\beta-1}{\beta}\right)$ is unique; indeed, observe that both $\alpha_1(\beta):= \log \beta$ and $\alpha_2(\beta):=3\cdot \left(\frac{\beta-1}{\beta}\right)$ are increasing with $\beta$. Further, we have $\alpha_1(3)<\alpha_2(3)$, with the derivatives obeying $\alpha_1'(\beta)<\alpha_2'(\beta)$, if and only if $\beta<3$, implying that $\alpha_1(\beta)<\alpha_2(\beta)$, for $\beta<3$. The above facts when put together show that the curves $\alpha_1(\cdot)$ and $\alpha_2(\cdot)$ cross each other exactly once, and this happens at a value $\beta = \beta^\star > 3$.
    
    Numerically, $\beta^\star$ can be estimated to be roughly $16.79$. Furthermore, when $\beta$ is constrained to be an integer, since $f(17)<f(16)$, we have that the minimizer $\beta^\star_\text{int}$ of $f$ over the integers, equals $17$.
\end{proof}
A couple of remarks are in order. Firstly, note from Lemma \ref{lem:optbranch} that the optimal (common) value of the branching factor $\beta^\star$ (equivalently, $\beta^\star_\text{int}$ over the integers) that minimizes the $\ell_{2}$-error of a level-uniform tree-based mechanism is \emph{independent} of the parameter $K$. Moreover, interestingly, $\beta^\star_\text{int} = 17$ is quite close to the ``optimal'' branching factor of $16$, which was obtained in \cite[Sec. 3.2]{qardaji} using somewhat heuristic analysis on a different error metric. Furthermore, we mention that the work \cite{qardaji} only considers those tree structures obtained by using a constant branching factor and privacy budgets at all levels of the tree; we \emph{explicitly prove} that such a constant branching factor and privacy budget is \emph{optimal} among all level-uniform tree-based mechanisms, for the $\ell_2$ error. In addition, our analysis also allows us to identify, via Theorem \ref{thm:e2treeopt1}, the optimal values of either $\mathbf{n}$ or $\boldsymbol{\varepsilon}$, when the value of the other is fixed. Some special cases such as geometrically distributed $\mathbf{n}$ and $\boldsymbol{\varepsilon}$ values were considered in \cite[Sec. IV]{cormode}.


Figure \ref{fig:plots} shows comparisons between the (normalized) $\ell_2$-errors of the mechanisms for approximate CDF estimation discussed until now, for the case when $K=256$; the errors are each multiplied by the common value $N^2$. We mention that for this (relatively small) value of $K$, the na\"ive histogram-based estimator $\mathbf{F}^\text{hist}$ performs better (has smaller $\ell_2$ error) compared to even the binary tree mechanism $\mathbf{F}^\text{tree}_{\mathbf{2},\frac{\varepsilon}{\log_2 K}\cdot \mathbf{1}}$. As expected (from the behaviour of $E_2^\text{OPT}$ in Lemma \ref{lem:optbranch}), the  $16$-ary level-uniform tree-based mechanism $\mathbf{F}^\text{tree}_{\mathbf{16},\frac{\varepsilon}{\log_{16} K}\cdot \mathbf{1}}$ with equal privacy budgets across all levels has the lowest $\ell_2$ error compared to the others plotted.

In Appendix \ref{sec:honaker}, with the aid of techniques from \cite{honaker}, we show that the level-uniform tree-based mechanism $\mathbf{F}^\text{tree}_{\beta\cdot \mathbf{1},\widehat{\varepsilon}\cdot \mathbf{1}}(\mathcal{D})$, with equal branching factors and equal privacy budgets across all tree levels, can be ``refined'' to yield a new mechanism $\mathbf{F}^\text{\normalfont tree, ref}_{\beta,\widehat{\varepsilon}}(\mathcal{D})$, with a much lower squared $\ell_2$-error, as given in Lemma \ref{lem:honaker}.

\begin{lemma}
	\label{lem:honaker}
	Given an integer $\beta = K^{1/(m-1)}$ and a positive real number $\widehat{\varepsilon} = \frac{\varepsilon}{m-1}$, we have
	\[
	E_2\left(\mathbf{F}^\text{\normalfont tree, ref}_{\beta,\widehat{\varepsilon}}(\mathcal{D})\right) = \frac{2K}{N^2}\cdot \left(\frac{\beta-1}{\widehat{\varepsilon}^2}\right)\cdot \sum_{i=1}^{m-1} \left(\sum_{j=0}^{m-1-i} \beta^{-j}\right)^{-1}.
	\]
\end{lemma}

\section{A Discussion on Integer-Valued Branching Factors}
\label{sec:integer}
In the previous section, we relaxed Definition \ref{def:cdftree} to allow for non-integer values of the parameters $n_1,\ldots,n_{m-1}$, in an attempt to obtain an understanding about the optimal squared $\ell_{2}$-error. We then showed that the optimal integer branching factor, when one allows for non-integer values, is $\beta^\star_\text{int} = 17$. However, typical values of $K$ do not admit a (common) branching factor of $17$ for any choice of {height $\overline{m}$} of the tree $\mathcal{T}_m$; here, we set $\overline{m}:= m-1$, for ease of reading. In this section, we discuss how one can ameliorate this issue, for the purpose of practical implementations.

Suppose that an integer $K$ representing the number of bins is given. Definition \ref{def:cdftree} then suggests a method to design level-uniform tree-based mechanisms for the fixed value of $K$: pick an integer {$\overline{m}\geq 1$}
and integer factors $n_i\geq 2$, $i\in [\overline{m}]$ of $K$ such that $\prod_{i\leq \overline{m}} n_i = K$. Further, let $\boldsymbol{\varepsilon} = \widehat{\varepsilon}\cdot \mathbf{1}$ be the length-$(\overline{m})$ vector, each of whose entries is $\widehat{\varepsilon} = \frac{\varepsilon}{\overline{m}}$, with $\sum_{i\leq \overline{m}} \varepsilon_i = \varepsilon$. Such a choice immediately gives rise to an $\varepsilon$-DP mechanism $\mathbf{F}_{\mathbf{n},\boldsymbol{\varepsilon}}$, where $\mathbf{n} = (n_1,\ldots,n_{\overline{m}})$. Furthermore, via Theorem \ref{thm:e2tree}, we have an explicit expression for the squared $\ell_2$-error $E_2(\mathbf{F}_{\mathbf{n},\widehat{\varepsilon}\cdot \mathbf{1}}(\mathcal{D}))$. One can then find an optimal choice of $\mathbf{n}$ for the given value of $K$, by searching over all collections of factors $\{n_i\}$ as above.

We shall now briefly discuss the complexity of such an exhaustive search-based procedure. Let $P_K$ be the multiset of prime factors of $K$, including multiplicities, with $|P_K|$ equalling the prime omega function $\Omega(K)$ (see sequence A001222 in \cite{oeis}). For any given height $\overline{m}$ of the tree $\mathcal{T}_{{m}}$, one method of counting the number of possible allocations of branching factors is as follows: consider the number of ways of choosing some $\overline{m}$ disjoint subsets of $P_K$, where $1\leq \overline{m}\leq \Omega(K)$, such that their union equals $P_K$. Some thought reveals that the number of possible allocations is precisely this quantity, which in turn {is at most $\Omega(K)^{\Omega(K)}:= C(K)$.}
While we have that ``on average'', $\Omega(K)$ behaves like $\log \log K$ \cite[Sec. 22.10--22.11]{hardywright}, indicating that a nai\"ve, explicit search over $C(K)$ allocations is possibly not computationally hard, this procedure still leaves open the question of obtaining the prime factors of $K$ -- a well-known hard problem. This hence motivates the need for explicit, optimal solutions to minimizing $E_2(\mathbf{F}_{\mathbf{n},\widehat{\varepsilon}\cdot\mathbf{1}}(\mathcal{D}))$ over integer-valued $\mathbf{n}$; we discuss this next for selected values of $K$. 

In particular, in what follows, we show that the optimal branching factors for a {fixed height $\overline{m}$}
 of the tree $\mathcal{T}_{{m}}$, for the case when $K = p^t$, $p$ being a prime, can be computed explicitly. We also demonstrate explicit values for, or bounds on, the optimal heights of a tree when $K = b^t$, for selected integers $b$. For general values of $K$, our strategy hence is to modify Definition \ref{def:cdftree} slightly, by fixing a constant $\overline{K}\geq K$ and considering branching factors $n_1,\ldots,n_{\overline{m}}$ (and integer $\overline{m}\geq 1$) such that $\prod_{i\leq \overline{m}} n_i = \overline{K}\geq K$. In particular, we choose $\overline{K}$ to be that power of a prime $p$ that is closest to, and at least, $K$. Let $\overline{K} = p^t$, for some prime $p\geq 2$, with $t\geq 1$.

Now, fix an {$\overline{m}\geq 1$} and the privacy budgets given by $\boldsymbol{\varepsilon} = \widehat{\varepsilon}\cdot \mathbf{1}$.
Via Theorem \ref{thm:e2tree}, our task of finding an optimal allocation (in terms of the $\ell_2$-error) $\mathbf{n}$ for the number of bins being $\overline{K}$, reduces to the following problem:
\begin{align}
	&{\text{minimize}}\quad {E}(\mathbf{n}):= \sum_{i=1}^{\overline{m}} n_i\notag\\
	&\text{subject to:}\ \  \prod_{i\leq \overline{m}} n_i = p^t,\ n_i\in \mathbb{N}. \notag
\end{align}
Now, since the only factors of $\overline{K} = p^t$ are integers of the form $p^{\tau}$, where $\tau\leq t$, we can let $n_i = p^{t_i}$, for each $i\in [\overline{m}]$, where $t_i\geq 1$. Let $\mathbf{t}:= (t_1,\ldots,t_{\overline{m}})$. The integer optimization problem above hence reduces to the following problem:
\begin{align}
	&{\text{minimize}}\quad \overline{E}(\mathbf{n}):= \sum_{i=1}^{\overline{m}} p^{t_i}\notag\\
	&\text{subject to:}\ \  \sum_{i\leq \overline{m}} t_i = t,\ t_i\in \mathbb{N}. \tag{O$^\text{int}$} \label{eq:optint1}
\end{align}
Note that, by symmetry, it suffices to focus on assignments $\mathbf{t}$ such that ${1 \le \ } t_1\leq t_2\leq\ldots \leq t_{\overline{m}}$. {Observe also that the constraint $\sum_{i\leq \overline{m}} t_i = t$, in turn,  implies $\overline{m} \le t$.} Let $\tau\geq 0$ and $0\leq \rho< \overline{m}$ be the unique integers such that $t = \tau\cdot \overline{m}+\rho$. The following claim then holds true.
\begin{theorem}
	\label{thm:intbranch}
	For any $\overline{m}\geq 2$, we have that if $\overline{m}$ divides $t$, then an optimal assignment $\mathbf{t}^\star$ for \normalfont{(O$^\text{int}$)} is
	\[
	\mathbf{t}^\star = \frac{t}{\overline{m}}\cdot \mathbf{1}.
	\]
	Else, the unique optimal assignment $\mathbf{t}^\star$ of non-decreasing integers is
	\[
	t_i^\star = \begin{cases}
		\tau,\ \text{for }1\leq i\leq \overline{m}-\rho,\\
		\tau+1,\ \text{for }m-\rho < i\leq \overline{m}.
	\end{cases}
	\]
\end{theorem}
\begin{proof}
	Consider first the case where $\overline{m}$ divides $t$. Observe that if one relaxes the requirement that $t_i\in \mathbb{N}$, then the optimization problem (O$^\text{int}$) is convex in the variables $\{t_i\}$. By standard arguments using the necessity of the KKT conditions \cite[Sec. 5.5.3]{boyd}, we obtain that the optimal assignment $\mathbf{t}^\star \in \mathbb{R}^{\overline{m}}$ satisfies $\mathbf{t}^\star = \frac{t}{\overline{m}}\cdot \mathbf{1}$. Since each coordinate of $\mathbf{t}^\star$ is an integer, we have that $\mathbf{t}^\star$ is also optimal for the \emph{integer} optimization problem (O$^\text{int}$).
	
	Next, consider the case when $t$ is not a multiple of $\overline{m}$. Consider first the special case when $\overline{m} = 2$, with $t$ being odd. In this case, it is easy to show by direct computation that the assignment $t_1 = \left \lfloor t/2 \right \rfloor$, $t_2 = 1+t_1 = \left \lceil t/2 \right \rceil$ is the optimal integer solution. Here, we have $\tau = \left \lfloor t/2 \right \rfloor$ and $\rho = 1$, verifying the statement of the theorem in this case.
	
	Now, to show the claim for {$\overline{m} \ge 3$}, assume the contrary, i.e., that there exists an optimal solution $\mathbf{t}^\star$ with $t_{\overline{m}}^\star>t_1^\star+1$. Now, consider the vector $\mathbf{t}'$ with $t'_{\overline{m}} = t^\star_{\overline{m}}-1$ and $t'_1 = t^\star_1+1$, with $t'_i = t^\star_i$, for $i\notin \{1,\overline{m}\}$. Note that we have 
	\[
	p^{t_1'}+p^{t_{\overline{m}}'} = p^{t^\star_1+1}+p^{t^\star_{\overline{m}}-1} < p^{t^\star_1}+p^{t^\star_{\overline{m}}},
	\]
	for $p\geq 2$. This indicates that $\mathbf{t}^\star$ is not optimal, hence leading to a contradiction. Therefore, we must have that $t^\star_{\overline{m}} = t^\star_1+1$, in this case. Now, let $r$ be such that $t_1^\star = t_2^\star = \ldots t_{\overline{m}-r}^\star = \tau$, and $t_{m-r}^\star = t_{m-r+1}^\star = \ldots = t_{\overline{m}}^\star = \tau+1$. Since we must have $\sum_{i\leq \overline{m}} t_i^\star = t$, it is clear that $r = \rho$, as in the statement of the lemma.
\end{proof}
We next discuss optimal choices of the {height $\overline{m}$} of the tree $\mathcal{T}_m$, given a positive integer $\overline{K}$. To this end, observe from Theorem \ref{thm:e2tree} that when $\varepsilon_i = \widehat{\varepsilon} = \frac{\varepsilon}{\overline{m}}$, for all $i\in [\overline{m}]$, we have that for any admissible allocation of branching factors $\mathbf{n}$ such that $\prod_{i\leq \overline{m}} n_i = \overline{K}$,
\[
E_2(\mathbf{F}_{\mathbf{n},\widehat{\varepsilon}\cdot \mathbf{1}}) = c\overline{m}^2\cdot \sum_{i\leq \overline{m}} (n_i-1),
\]
for some constant $c>0$. Let us define $g_{\overline{m}}(\mathbf{n}):= \overline{m}^2\cdot \sum_{i\leq \overline{m}} (n_i-1)$. The following lemma thus holds.
\begin{lemma}
	\label{lem:helperint}
	We have that for all $\mathbf{n}$,
	\[
	g_{\overline{m}}(\mathbf{n})\geq g_{\overline{m}}(\overline{K}^{1/\overline{m}}\cdot \mathbf{1}).
	\]
\end{lemma}
\begin{proof}
	Observe that
	\begin{align*}
		g_{\overline{m}}(\mathbf{n})&= \overline{m}^2\cdot \sum_{i\leq \overline{m}} (n_i-1)\\
		&=\overline{m}^3\cdot \left(\frac{\sum_{i\leq \overline{m}} n_i}{\overline{m}}-1\right)\\
		&\geq \overline{m}^3\cdot  \left(\left({\prod_{i\leq \overline{m}} n_i}\right)^{1/\overline{m}}-1\right)\\ &= \overline{m}^3\cdot (K^{1/\overline{m}}-1) = g_{\overline{m}}(\overline{K}^{1/\overline{m}}\cdot \mathbf{1}).
	\end{align*}
Here, the inequality is a simple consequence of the AM-GM inequality.
\end{proof}
Now, we define $$h(\overline{m}):= \min_{\substack{\mathbf{n}:\ \prod_{i\leq \overline{m}}n_i = \overline{K},\\\mathbf{n}\in \mathbb{N}^{\overline{m}}}} g_{\overline{m}}(\mathbf{n}).$$ From Lemma \ref{lem:helperint}, it is clear that $h(\overline{m})\geq g_{\overline{m}}(\overline{K}^{1/\overline{m}}\cdot \mathbf{1})$, with  equality if $K^{1/\overline{m}}$ is an integer. Now, let us define
\begin{equation}
	\label{eq:moptdef}
	\overline{m}_\text{OPT} := \argmin\limits_{\overline{m}\in \mathbb{N}} h(\overline{m}).
\end{equation}
Further, let
\[
f(\beta):= \frac{\beta-1}{(\log \beta)^3}.
\]
Recall from Lemma \ref{lem:optbranch} that for a fixed number of bins $\overline{K}$, the function $f$ characterizes the squared $\ell_2$ error $E_2^\text{OPT}$ of a tree with equal branching factors (and privacy budgets) at all levels, with $\beta$ denoting the common branching factor. In what follows, we present bounds on $\overline{m}_\text{OPT}$ for cases when $\overline{K} = \alpha^t$, for selected integers $\alpha$. In particular, we exactly identify the optimal height of the tree $\mathcal{T}_m$ for the case when $\alpha\geq 17$ and is prime.

\begin{theorem}
	\label{thm:helperint1}
	For any integer $\tilde{m}\in \mathbb{N}$ such that $b:= \overline{K}^{1/\tilde{m}}$ is an integer, with $b\geq 17$, we must have $\overline{m}_\text{\normalfont OPT}\geq \tilde{m}$.
	
	Further, for any integer $\tilde{m}\in \mathbb{N}$ such that $b:= \overline{K}^{1/\tilde{m}}$ is an integer, with $b\leq 16$, we must have {$\overline{m}_\text{\normalfont OPT}\leq \tilde{m}$}.
\end{theorem}
\begin{proof}
	We shall first prove the first statement in the theorem. Let $\widehat{m}$ be any integer such that {$\widehat{m} < \tilde{m}$}. It suffices to show that $h(\widehat{m})> h(\tilde{m})$. Indeed, observe that
	\begin{align*}
		h(\widehat{m})&\geq g_{\widehat{m}}(K^{1/\widehat{m}}\cdot \mathbf{1})\\
		&= { \left(\log K\right)^3 }\cdot f(K^{1/\widehat{m}})\\
		&> { \left(\log K\right)^3 }\cdot f(b) = g_{\tilde{m}}(b\cdot \mathbf{1}) = h(\tilde{m}).
	\end{align*}
Here, the second (strict) inequality follows from the properties of the function $f$ discussed in the proof of Lemma \ref{lem:optbranch}.

Likewise, to prove the second statement in the theorem, it suffices to show that $h(\widehat{m})> h(\tilde{m})$, for any integer $\widehat{m}>\tilde{m}$. It can be verified that the same sequence of inequalities above holds, due to the properties of the function $f$ discussed in the proof of Lemma \ref{lem:optbranch}.
\end{proof}
As direct corollaries of Theorem \ref{thm:helperint1}, we obtain the following statements.
\begin{corollary}
	\label{cor:int1}
	Let $\overline{K} = p^t$, for some integer $t\geq 1$ and prime $p\geq 17$. Then, we have $\overline{m}_\text{\normalfont OPT} = t$.	Furthermore, for $\overline{K}\leq 16$, we have $\overline{m}_\text{\normalfont OPT} = 1$.
\end{corollary}

\begin{corollary}
	\label{cor:int2}
	Suppose that $\overline{K} = q^t$, for integers $q, t\geq 1$. Let $\overline{m}_1, \overline{m}_2$ be integers with $b_1:= \overline{K}^{1/\overline{m}_1}$ and $b_2:= \overline{K}^{1/\overline{m}_2}$. If $b_1$ and $b_2$ are both integers that satisfy $b_1\geq 17$ and $b_2\leq 16$, then we have $\overline{m}_1\leq \overline{m}_\text{\normalfont OPT}\leq \overline{m}_2$.
\end{corollary}

As example applications of Corollary \ref{cor:int2}, consider the case when $\overline{K} = 2^{20}$. Choosing $\overline{m}_1 = 4$ and $\overline{m}_2 = 5$, we see that $b_1 = \overline{K}^{1/\overline{m}_1} = 2^5>17$. Further, $b_2 = \overline{K}^{1/\overline{m}_2} = 2^4 = 16$. Hence, we have that $\overline{m}_\text{OPT}\in \{4,5\}$.

Likewise, consider the situation where $\overline{K} = 3^{6t}$, for some integer $t$. By choosing $\overline{m}_1 = 2t$ and $\overline{m}_2 = 3t$, we see that $b_1 = 3^3>17$ and $b_2 = 3^2<16$, implying that $\overline{m}_\text{OPT}\in [2t,3t]$, in this case.

While Theorem \ref{thm:e2treeopt1}  indicates that for mechanisms with equal privacy budget allocations across tree levels, we must have equal branching factors across levels too to achieve the smallest possible $\ell_2$-error, an analogous argument \emph{does not} hold in the case when the branching factors are restricted to be integers. For a fixed integer $\overline{K}$ and positive real $\varepsilon$, let $\mathbf{F}^\text{tree, eq}_{\mathbf{n},\varepsilon}$ denote any level-uniform tree-based mechanism with equal privacy budgets across all the tree levels. The following lemma, proved in Appendix \ref{app:intnoncommon}, then holds.

\begin{lemma}
	\label{lem:intnoncommon}
	For $\overline{K} = 2^t$, where $t\geq 11$ is prime, and any $\varepsilon>0$, there exists a level-uniform tree-based mechanism $\mathbf{F}^\text{\normalfont tree, eq}_{\mathbf{n}}$ with not-all-equal branching factors such that
	\[
	E_2\left(\mathbf{F}^\text{\normalfont tree,  eq}_{\mathbf{n}}\right)<\frac13\cdot E_2\left(\mathbf{F}^\text{\normalfont tree, eq}_{\beta\cdot \mathbf{1}}\right),
	\]
	for all admissible integers $\beta\geq 2$.
\end{lemma}
We end this section with a remark. As argued in Lemma \ref{lem:intnoncommon}, when the branching factors are restricted to be integers, for selected values of $\overline{K}$, the optimal branching factors in terms of $\ell_2$-error are not all equal across levels, when the privacy budgets are equal across levels. Given such a choice of {integer} branching factors, one can then refine the $\ell_2$-error further by choosing privacy allocations that are level-uniform, but potentially unequal-across-levels, via Theorem \ref{thm:e2treeopt1}.
\section{Post-Processing CDFs for Consistency}
\label{sec:post-process}

The material in Sections \ref{sec:tree} and \ref{sec:integer} provided us with an understanding of the structure(s) of  level-uniform tree-based mechanisms for obtaining low $\ell_2$-error. In this section, we demonstrate that it is possible in practice to reduce the error even further, by suitably post-processing the noisy CDF estimates to comply with the properties of a CDF.

In particular, we specify a natural \emph{consistency} requirement on any estimate of a CDF; the term ``consistency" here simply refers to standard desirable properties of any statistic of interest. We mention that several variants of post-processing private \emph{histogram} estimates for consistency have been considered in the literature 
\cite{barak,histhay,histml} and implemented in practice in the 2020 United States Census \cite{topdown}. Of immediate relevance are the work \cite{histhay}, which defined and explicitly solved an instance of a ``tree-based'' consistency problem for the $\ell_2$-error, and the work \cite{histml}, which provided iterative, numerical algorithms for tree-based consistency under the $\ell_1$-error. The output of such a tree-based consistency problem is a vector of ``consistent'' estimates of the counts at each node of the tree of interest.

While the constraints imposed via consistency extend naturally from private trees for histograms to private trees for CDFs, we take a simpler, ``first-principles'' approach to post-processing CDF estimates. Specifically, given any vector of (noisy) values that represents a CDF estimate, obtained via a tree-based approach or otherwise, we demonstrate a simple, dynamic programming-based algorithm for post-processing to yield consistent CDF estimates under \emph{any} error metric that is additive across the indices of the vector. In addition to facilitating downstream processing of the private CDF estimate, such post-processing typically also results in decreased errors.



\subsubsection{The CDF Consistency Problem}
Fix a number of bins $K\geq 1$ and any error metric $\Delta: \mathbb{R}^K\times \mathbb{R}^K\to \mathbb{R}$ such that for length-$K$ vectors $\mathbf{u}, \mathbf{v}$, we have
\[
\Delta(\mathbf{u},\mathbf{v}) = \sum_{i=1}^K d(u_i,v_i),
\]
for some function $d: \mathbb{R}^2\to \mathbb{R}$. Examples of such error metrics include the $\ell_1$- and $\ell_2$-error metrics (and in general, the $\ell_p$-error metric for any $p\geq 0$), the Hamming error metric, and so on. Let $\widehat{\mathbf{h}}(\mathcal{D}):= \widehat{\mathbf{h}} = (\widehat{h}_1,\ldots,\widehat{h}_K)$ be any vector that represents a potentially noisy estimate of the cumulative counts $\overline{\textbf{c}}$ defined in Section \ref{sec:prelim}. For example, $\widehat{\mathbf{h}}$ could be the vector $N\cdot \mathbf{F}_{\mathbf{n},\boldsymbol{\varepsilon}}^\text{tree}(\mathcal{D})$ in Definition \ref{def:cdftree}. We assume that $\widehat{h}_K = N$, since the total number of samples is publicly known. Since the cumulative count vector $\overline{\textbf{c}}$ is integer-valued, non-negative, and non-decreasing, with $\overline{c}_K = N$, we wish to ``optimally'' post-process $\widehat{\mathbf{h}}$ into such a consistent estimate $\mathbf{h}$ that minimizes $\Delta(\mathbf{h},\widehat{\mathbf{h}})$. In other words, our optimization problem of interest is:
\begin{align}
	&{\text{minimize}}\quad {\Delta}(\mathbf{h},\widehat{\mathbf{h}})\notag\\
	&\text{subject to:}\ \  \mathbf{h}\in \mathbb{N},\ \mathbf{h}\geq 0,\ h_K = N,\ h_i\leq h_{i+1}, \notag\\
	&\ \ \ \ \ \ \ \ \  \ \ \ \ \ \  \ \ \ \ \ \ \ \ \ \ \ \   \ \ \ \ \  \ \ \text{ for all $1\leq i\leq K-1$}. \label{eq:optcons} \tag{O$^c$}
\end{align}
An optimal, consistent CDF $\mathbf{F}^\text{cons}_\Delta$ can then be obtained from an optimal solution $\mathbf{h}^\text{OPT}$ to \eqref{eq:optcons} by simply setting $\mathbf{F}^\text{cons}_\Delta = \frac{1}{N}\cdot \mathbf{h}^\text{OPT}$. 
\subsubsection{A Dynamic Programming-Based Solution}
It turns out that the problem (\ref{eq:optcons}) has a simple solution that can be obtained via dynamic programming. To this end, we employ a trellis -- a directed, acyclic graph widely used in communication theory and error-control coding (see \cite{vardysurvey,viterbi} for more on the use of trellises for decoding over noisy communication channels). A trellis serves as a convenient data structure for our problem; in what follows, we briefly describe its structure. The trellis that we construct is a directed graph $\mathcal{G} = (\mathcal{V},\mathcal{E})$, whose vertex set $\mathcal{V}$ is partitioned into $K+1$ ``stages'' $\mathcal{V}_0, \mathcal{V}_1,\ldots,\mathcal{V}_K$, where $\mathcal{V}_0 = \{\textsf{root}\}$ and $\mathcal{V}_K = \{\textsf{toor}\}$. We set $\mathcal{V}_i = \{0,1,\ldots,N\}$, for $1\leq i\leq K-1$. Further, any (directed) edge $(v,v')$ in $\mathcal{G}$ begins at $v\in \mathcal{V}_i$ and ends at $v'\in \mathcal{V}_{i+1}$; further, we only allow edges $(v,v')$ where $v\leq v'$, with  $v,v'\in \{0,1,\ldots,N\}$, for each $i\in [K-1]$, and attach \textsf{root} to all nodes in $\mathcal{V}_1$ and all the nodes in $\mathcal{V}_{K-1}$ to  \textsf{toor}. Hence, the edge set $\mathcal{E}$ can also be partitioned into disjoint sets $\mathcal{E}_0, \mathcal{E}_1,\ldots, \mathcal{E}_K$, where $\mathcal{E}_i$ contains all admissible edges from stage $i-1$ to $i$, for $i\in [K]$. Clearly, the trellis $\mathcal{G}$ represents, via the labels of vertices along its paths, the collection of integer-valued, non-negative, non-decreasing sequences, each of whose coordinates is in $\{0,1,\ldots,N\}$. Further, classical results \cite{vardykschischang} confirm that $\mathcal{G}$ is the (unique) ``minimal'' trellis, possessing the fewest overall number of edges and vertices and having the smallest complexity of the dynamic programming method we shall describe next.

Given such a trellis $\mathcal{G}$, we assign to any edge $(v,v')$, where $v\in \mathcal{V}_i$ and $v'\in \mathcal{V}_{i+1}$, for some $i\in [0:K-1]$, the weight $w(v,v') = d(v,\widehat{h}_i)$. For example, using the $\ell_1$-error metric, we would have $w(v,v') = |v-\widehat{h}_i|$ and using the squared $\ell_2$-error metric, we would have $w(v,v') = (v-\widehat{h}_i)^2$. Now, we construct a vector $$\alpha:=\left(\alpha(v,i): i\in [0:K-1],\ v\in \mathcal{V}_i\right),$$  
indexed by tuples $(v,i)$ of (node, stage) pairs as follows. We set $\alpha(\textsf{toor},K) = 0$ and for $i\in [0:K-1]$,
\begin{equation}
	\label{eq:cons1}
\alpha(v,i) = \min_{v':\ (v,v')\in \mathcal{E}_{i+1}} \{w(v,v')+\alpha(v',i+1)\}.
\end{equation}
Furthermore, for each $(v,i)$ pair, we store that node $\iota(v,i)\in \mathcal{V}_{i+1}$ that attains the minimum above. Finally, we set
\begin{equation}
	\label{eq:cons2}
h_0:= \textsf{root}\ \text{and}\ h_K = N,
\end{equation}
and for $i\in [K-1]$, we assign
\begin{equation}
	\label{eq:cons3}
h_i = \iota(h_{i-1},i-1).
\end{equation}
It is clear from the sub-problem structure in \eqref{eq:optcons} that the above dynamic programming-based approach yields the optimal solution $\mathbf{h} = (h_1,\ldots,h_{K})$. Furthermore, the time and space complexity of this approach is $O(K\cdot N^2)$ (see \cite{vardykschischang} for a more careful handling of the complexity of such trellis-based dynamic programming approaches), which allows for relatively fast implementations in practice.

\subsubsection{Numerical Experiments}
Given the solution structure in \eqref{eq:cons1}--\eqref{eq:cons3}, we next run experiments to numerically verify the efficacy of the post-processing procedure. To this end, we set $K = 997$ and $N = 900$ and use a level-uniform tree-based mechanism with tree height equalling $1$; in other words, we use the mechanism $\mathbf{F}^\text{tree}_{K,\varepsilon} = \mathbf{F}^\text{hist}$. This is reasonable since from Corollary \ref{cor:int1}, we have that the optimal height under the $\ell_{2}$-error is indeed $1$, since $K$ is prime. 

We first sample $N$ uniformly random points in the interval $[a,b) = [0,K)$ and then compute the true counts $\chi(v)$ corresponding to each node $v$ with associated interval ${\mathcal{I}}(v) = \left[\frac{j-1}{K},\frac{j}{K}\right)$, for $j\in [K]$. This then allows us to compute the true CDF $\mathbf{F}$ pertaining to this dataset of samples. We then compute the noisy CDF $\mathbf{F}^\text{tree}_{K,\varepsilon}$ via \eqref{eq:mchi}, \eqref{eq:mchilap}, and Definition \ref{def:cdftree}. Next, we set $\widehat{h} = N\cdot \mathbf{F}^\text{tree}_{K,\varepsilon}$ and compute  optimal, consistent solutions $\mathbf{F}^\text{cons}_{\ell_2}$ and $\mathbf{F}^\text{cons}_{\ell_1}$ corresponding to the error metric $\Delta$ being the $\ell_2$- and $\ell_1$-error metrics, respectively. 
In addition to the apparent advantage of consistency that is enforced for downstream applications, the post-processing procedure also results in smaller errors. For example, the errors shown below are the empirical average errors incurred in $100$ Monte-Carlo simulations of the above procedure, when $\varepsilon = 0.1$. 
\[
\left \lVert \mathbf{F} - \mathbf{F}_{K,\varepsilon}^\text{tree}\right \rVert_1 \approx 502.81 \ \text{and}\ \left \lVert \mathbf{F} - \mathbf{F}_{\ell_1}^\text{cons}\right \rVert_1 \approx 286.43,\ \text{and}
\]
\[
\left \lVert \mathbf{F} - \mathbf{F}_{K,\varepsilon}^\text{tree}\right \rVert_2 \approx 18.54 \ \text{and}\ \left \lVert \mathbf{F} - \mathbf{F}_{\ell_2}^\text{cons}\right \rVert_2 \approx 10.72.
\]

\section{Conclusion}
\label{sec:conclusion}
In this paper, we defined and analyzed the broad class of level-uniform tree-based mechanisms for the $\varepsilon$-differentially private release of (approximate) CDFs. In particular, when the branching factors are relaxed to take on real values, we showed that for a fixed number of bins $K$, a simple, optimal strategy exists for matching either of the branching factors or the privacy budgets across tree levels, when the other is fixed. Furthermore, in this relaxed setting, if each of the branching factors is at least $3$, the overall optimal level-uniform tree-based mechanism, for sufficiently large number of bins, has equal branching factors (roughly {$17$}) and equal privacy budgets across all levels. Furthermore, when the integer constraint on branching factors is introduced, we discussed the structure of the optimal tree-based mechanisms, {when $K$ is a prime power. Finally, we discussed a general strategy} to further reduce errors by optimally post-processing the noisy outputs of the mechanism for consistency.

{In light of Lemma~\ref{lem:intnoncommon},} an interesting area of future work is the application of level-uniform tree based mechanisms, with potentially different branching factors at different tree levels, to the release of wider classes of statistics, such as those in \cite{decayedsums, windowed}. 





%
\bibliographystyle{IEEEtran}
{\footnotesize
	\bibliography{references}}

\appendices
\section{Proof of convexity of $\Phi$}
\label{sec:app-convex}
In this section, we prove that the function
\[
\Phi(\mathbf{r},\boldsymbol{\varepsilon})= \sum_{i=1}^{m-1} \frac{e^{r_i}-1}{{\varepsilon_i}^2}
\]
is jointly convex in $\mathbf{r}, \boldsymbol{\varepsilon}$.
\begin{proof}
	It suffices to show that the function
	\[
	\phi(r,\varepsilon):= \frac{e^{r}-1}{\varepsilon^2}
	\]
	is convex in $(r,\varepsilon)$, for $r\geq \log 3$ and $\varepsilon\geq 0$. The result then follows, since the sum of convex functions is convex.
	
	Now, observe that the Hessian of $\phi$ is
	\[
	H = \varepsilon^{-2}\cdot\begin{bmatrix}
		e^r & {-2e^r\varepsilon^{-1}}\\
		{-2e^{r}\varepsilon^{-1}} & {6\varepsilon^{-2}(e^r-1)}
	\end{bmatrix}.
	\]
	Since det$(H)\geq 0$ and trace$(H)> 0$, for $r\geq \log 3$ and $\varepsilon\geq 0$, we obtain that the eigenvalues of $H$ are both positive, implying that $H$ is positive definite. This, in turn, implies that $\phi$ is convex in admissible $(r,\varepsilon)$.
\end{proof}
\section{Improving $\ell_2$-Error Using Multiple Noisy Views}
\label{sec:honaker}
In this section, we briefly discuss a strategy for refining the noisy CDF estimate $\mathbf{F}^\text{tree}_{\beta\cdot \mathbf{1}, \widehat{\varepsilon}\cdot \mathbf{1}}$ obtained via a tree-based mechanism with equal branching factors (each equalling $\beta$) and equal privacy budgets (each equalling $\widehat{\varepsilon}$), by employing the approach used in \cite{honaker}. For this discussion, we fix a height $m-1$ of the tree $\mathcal{T}_m$ of interest. Now, the vector of estimates $M_{\chi}^\text{Lap}(\mathcal{D})$ (see \eqref{eq:mchilap}) is refined in two steps (we drop the explicit dependence on $\mathcal{D}$ for ease of reading):
\begin{enumerate}
	\item The noisy estimates $M_{\chi(u)}^\text{Lap}$ of counts corresponding to child nodes $u$ of a node $v$ in $\mathcal{T}_m$ are used to refine the estimate $M_{\chi(v)}^\text{Lap}$, and
	\item For each leaf node $v$, the estimates $\overline{M}_{\chi,v}^\text{Lap}$ of the cumulative counts corresponding to node $v$, obtained \emph{after} the refinement procedure in Step 1), are further refined using an estimate of the count in the interval $[a,b)\setminus \overline{\mathcal{I}}(v)$.
\end{enumerate}
We mention that Step 1) is precisely the ``estimation from below" approach in \cite[Sec. 3.2]{honaker}, and the composition of the two steps above is closely related, but suboptimal when compared to, the ``fully efficient estimation'' in \cite[Sec. 3.4]{honaker}. However, via the two steps above, we can obtain a clean expression for the squared $\ell_2$-error post the refinement procedure. Such a refinement can in practice then be followed by the post-processing method for consistency discussed in Section \ref{sec:post-process}. In what follows, we briefly present the improvement in squared $\ell_2$-error obtained via this refinement; for a more detailed exposition, we refer the reader to \cite{honaker}.

\subsubsection{Elaborating on Step 1)} In this step, we use the fact that $M_{\chi(v)}^\text{Lap}$ and $\sum_{u:\ u\in \psi(v)} M_{\chi(u)}^\text{Lap}$, where $\psi(v):= \{u:\ u\text{ is a child of }v\}$ are two estimates of the same quantity $\chi(v)$. Thus, for every non-root node $v$, we can recursively construct new unbiased estimates \begin{equation} \label{eq:hon1}\tilde{M}_{\chi(v)} = \alpha\cdot M_{\chi(v)}^\text{Lap}+(1-\alpha)\cdot \sum_{u:\ u\in \psi(v)} \tilde{M}_{\chi(u)}\end{equation} 
	of $\chi(v)$, where $\alpha$ is chosen so as to minimize the variance of $\tilde{M}_{\chi(v)}$; here, we set $\tilde{M}_{\chi(v)} = M_{\chi(v)}^\text{Lap}$, when $v$ is a leaf node. Further, as before, we simply set $\tilde{M}_{\chi(v_1)}^\text{Lap} = N$, for the root node $v_1$. Recall that Var$\left(M_{\chi(v)}^\text{Lap}\right) = \frac{8}{\widehat{\varepsilon}^2}$, for all (non-root) nodes $v$ in $\mathcal{T}_m$.
	\begin{proposition}
		\label{prop:hon}
		For any node $v$ at level $i\in [m-1]$ of $\mathcal{T}_m$, we have
		\begin{align*}
		\text{\normalfont Var}\left(\tilde{M}_{\chi(v)}\right) &= \left(\sum_{j=0}^{m-1-i} \beta^{-j}\right)^{-1}\cdot \frac{8}{\widehat{\varepsilon}^2}.
		\end{align*}
	\end{proposition}
\begin{proof}
	The claim holds by induction. Indeed, the base case when $i = m-1$ is trivial. Now, assume that the claim holds for all levels $j\in \{i+1,\ldots,m-1\}$. From \cite[Eq. (4)]{honaker}, we see that the optimal choice of $\alpha$ in \eqref{eq:hon1} is $$\alpha = \frac{\text{\normalfont Var}\left(M_{\chi(v)}^\text{Lap}\right)^{-1}}{\text{\normalfont Var}\left(M_{\chi(v)}^\text{Lap}\right)^{-1}+\beta^{-1}\cdot \text{Var}\left(\tilde{M}_{\chi(u)}\right)^{-1}},$$ where $u$ is any child node of $v$. Furthermore, we have from \cite[Remark 2]{honaker} that for this choice of $\alpha$,
	\begin{align*}
		\text{\normalfont Var}\left(\tilde{M}_{\chi(v)}\right)  = \left(\text{\normalfont Var}\left(M_{\chi(v)}^\text{Lap}\right)^{-1}+\beta^{-1}\cdot \text{Var}\left(\tilde{M}_{\chi(u)}\right)^{-1}\right)^{-1}.
	\end{align*}
	Making use of the induction hypothesis and the fact that Var$\left(M_{\chi(v)}^\text{Lap}\right) = \frac{8}{\widehat{\varepsilon}^2}$ yields the statement of the proposition.
\end{proof}
\subsubsection{Elaborating on Step 2)} With the aid of the vector $\tilde{M}_{\chi}:= (\tilde{M}_{\chi(v)}:\ v\text{ is a node of $\mathcal{T}_m$})$, we construct, similar to \eqref{eq:cdfcompute}, a vector $\overrightarrow{M}_{\chi}:= (\overrightarrow{M}_{\chi, v}:\ v\text{ is a leaf node of $\mathcal{T}_m$})$, where
\begin{equation}
	\label{eq:mright}
\overrightarrow{M}_{\chi, v} = \sum_{v^\prime\in \mathcal{C}(v)} \tilde{M}_{\chi(v^\prime)}.
\end{equation}
Now, for any leaf node $v$, we associate the "right cumulative" interval $\overleftarrow{\mathcal{I}}(v):= [a,b)\setminus \overrightarrow{\mathcal{I}}(v)$. Similar to the covering $\mathcal{C}(v) = \mathcal{C}(\overrightarrow{\mathcal{I}})$ defined in Section \ref{sec:tree}, we define the "right covering" $\overline{\mathcal{C}}(v)$ by the following iterative procedure. We first set $\overline{\mathcal{J}} = \overleftarrow{\mathcal{I}}$ and $\overline{\mathcal{C}}(v)$ to be the empty set. Next, we start at level $i=0$ and identify the set $\mathcal{W}_i$ consisting of all nodes $u$ at level $i$ such that $\cup_{u\in \mathcal{W}_i} \mathcal{I}(u) \subseteq \overline{\mathcal{J}}$ and update $\overline{\mathcal{C}}(v)\gets \overline{\mathcal{C}}(v)\cup \mathcal{W}_i$. Then, we update $\overline{\mathcal{J}}\gets \overline{\mathcal{J}}\setminus \cup_{u\in \mathcal{W}_i} \mathcal{I}(u)$ and $i\gets i+1$, and iterate until $\overline{\mathcal{J}}$ is empty. Similar to Lemma \ref{lem:tau}, it is possible to explicitly identify the right covering $\overline{\mathcal{C}}(\overleftarrow{\mathcal{I}}) = \overline{\mathcal{C}}(v)$. For any $i\in [m-1]$, we write $v_{1,j_1,\ldots,j_{i-1},>j_{i}}$ as shorthand for the nodes $v_{1,j_1,\ldots,j_{i-1},j_i+1},\ldots, v_{1,j_1,\ldots,j_{i-1},n_i}$.
\begin{lemma}
	Let $\overleftarrow{\mathcal{I}}$ be a right cumulative interval associated with a leaf node $v = v_{1,j_1,\ldots,j_{m-1}}$. If $\overleftarrow{\mathcal{I}} = [a,b)$, then $\overline{\mathcal{C}}(\overleftarrow{\mathcal{I}})$ is the empty set. Else, 
	\[
	\overline{\mathcal{C}}(\overleftarrow{\mathcal{I}}) = \{v_{1,>j_1},v_{1,j_1,>j_2},\ldots,v_{1,j_1,j_2,\ldots,>j_{\tau}}\},
	\]
	where $\tau = \tau(v)$.
\end{lemma}
Finally, we define the length-$K$ vector $\overleftarrow{M}_{\chi} = (\overleftarrow{M}_{\chi, v}:\ \text{$v$ is a leaf node})$, such that
\begin{equation}
	\label{eq:mleft}
	\overleftarrow{M}_{\chi, v} = \sum_{v^\prime\in \overline{\mathcal{C}}(v)} \tilde{M}_{\chi(v^\prime)}.
\end{equation}
Observe that for any leaf node $v$, the random variables $\overrightarrow{M}_{\chi, v}$ and $(N- \overleftarrow{M}_{\chi, v})$ are two estimates of the same quantity. We then construct the estimate
\begin{equation}
	\label{eq:honav}
M^\text{ref}_{\chi, v}:= \frac{\overrightarrow{M}_{\chi, v}+N- \overleftarrow{M}_{\chi, v} }{2}
\end{equation}
as a refined estimate of the cumulative count at leaf node $v$. Finally, we set
\[
\mathbf{F}^\text{tree, ref}_{\beta,\widehat{\varepsilon}}(\mathcal{D}) = \frac{1}{N}\cdot M^\text{ref}_{\chi},
\]
where $M^\text{ref}_{\chi} = \left(M^\text{ref}_{\chi, v}:\ v\text{ is a leaf node of $\mathcal{T}_m$}\right)$. The procedure described in this section gives rise to Lemma \ref{lem:honaker}, in a manner entirely analogous to Theorem \ref{thm:e2tree}.

\begin{proof}[Proof of Lemma \ref{lem:honaker}]
	The proof is a direct consequence of Proposition \ref{prop:hon} and the fact that the averaging in \eqref{eq:honav} leads to a factor of half reduction in {the} overall $\ell_2$-error.
\end{proof}

\section{Proof of Lemma \ref{lem:intnoncommon}}
\label{app:intnoncommon}
In this section, we prove Lemma \ref{lem:intnoncommon}.

\begin{proof}
	Since $\overline{K} = 2^t$ is such that $t$ is prime, we must have that $\beta\in \{2,2^t\}$, where $\beta$ is any common branching factor, as in the statement of the lemma. 
	
	Now, since $t\geq 11$ is prime, we must have that $t=4s+1$ or $t=4s+3$, for some integer $s$.
	
	Consider the first case when $t=4s+1$, for some $s$. From the expression for the $\ell_2$-error in Theorem \ref{thm:e2tree}, we see that for fixed $\overline{K}$, 
	\[
	E_2\left(\mathbf{F}^\text{tree, eq}_{2\cdot \mathbf{1}}\right) = c\cdot t^3 = c\cdot (4s+1)^3> c\cdot (64s^3+48s^2),
	\]
	for some fixed $c>0$. Further, we have 
	\[
	E_2\left(\mathbf{F}^\text{tree, eq}_{\overline{K}}\right) = c\cdot (\overline{K}-1) > E_2\left(\mathbf{F}^\text{tree, eq}_{2\cdot \mathbf{1}}\right),
	\]
	for $t\geq 11$. Now, consider a level uniform tree-based mechanism $\mathbf{F}^\text{tree, eq}_{\overline{\mathbf{n}}}$ of height $\overline{m} = s$, with $\overline{n}_1 = \ldots = \overline{n}_{\overline{m}-1} = 16$ and $\overline{n}_{\overline{m}} = 32$. A direct computation then reveals that
	\begin{align*}
	E_2\left(\mathbf{F}^\text{tree, eq}_{\overline{\mathbf{n}}}\right) &= cs^2\cdot (15(s-1)+31)\\ &= c\cdot (15s^3+16s^2)\\
	&<\frac{c}{3}\cdot (64s^3+48s^2) = E_2\left(\mathbf{F}^\text{tree, eq}_{2\cdot \mathbf{1}}\right).
	\end{align*}
Consider next the case when $t = 4s+3$. Here, we have
\[
E_2\left(\mathbf{F}^\text{tree, eq}_{2\cdot \mathbf{1}}\right) = c_1\cdot (4s+3)^3>c_1\cdot (64s^3+144s^2),
\]
for some fixed $c_1>0$. Further, we again have that $E_2\left(\mathbf{F}^\text{tree, eq}_{\overline{K}}\right) > E_2\left(\mathbf{F}^\text{tree, eq}_{2\cdot \mathbf{1}}\right)$. Next, consider a level uniform tree-based mechanism $\mathbf{F}^\text{tree, eq}_{\tilde{\mathbf{n}}}$ of height $\overline{m} = s+1$, with $\tilde{n}_1 = 8$ and $\tilde{n}_2 = \ldots = \tilde{n}_{\overline{m}} = 16$. Then,
 \[
 E_2\left(\mathbf{F}^\text{tree, eq}_{\tilde{\mathbf{n}}}\right) = c_1(s+1)^2\cdot (15s+7)<c_1\cdot(16s^3+48s^2),
 \]
 implying that $E_2\left(\mathbf{F}^\text{tree, eq}_{\tilde{\mathbf{n}}}\right) <\frac13 \cdot E_2\left(\mathbf{F}^\text{tree, eq}_{2\cdot \mathbf{1}}\right)$.
\end{proof}
\end{document}